\documentclass[conference]{llncs}

\usepackage[english]{babel}
\usepackage{cancel} 

\usepackage{xifthen}
\usepackage{xparse}

\usepackage{biblatex}
\PassOptionsToPackage{hyphens}{url}\usepackage{hyperref}

\usepackage{footmisc}

 \usepackage{amsmath}
 \usepackage[misc]{ifsym}
 \usepackage{amssymb}
 \usepackage{pifont}
\usepackage{mathtools}
\usepackage{proof} 
\usepackage{stmaryrd} 
\usepackage{upgreek} 
\usepackage{stackengine}
\usepackage{adjustbox}
\usepackage{enumitem}
\usepackage{float}

\usepackage{tikz}
\usepackage[plain]{fancyref}
\usepackage{float}
\usepackage{subcaption}
\usepackage{wrapfig}

\usepackage{graphicx}

\usepackage{fancyhdr}
\usepackage{marginfix} 
\usepackage{microtype}
\usepackage{pdflscape} 
\usepackage{totpages}

\usepackage{longtable}
\usepackage{tabularx}
\usepackage{multirow}
\usepackage{booktabs}
\usepackage{adjustbox}
\usepackage{array}

\usepackage{multicol} 
\usepackage{enumitem}

\usepackage{listings}
\usepackage{algorithm}
\usepackage[noend]{algpseudocode}
\MakeRobust{\Call}

\usetikzlibrary{%
arrows,%
calc,%
shapes.misc,%
shapes.arrows,%
chains,%
matrix,%
positioning,%
scopes,%
decorations.pathmorphing,%
decorations.pathreplacing,%
shadows,%
er,%
trees,%
shapes,%
automata,%
patterns%
}

\pgfdeclarelayer{edgelayer}
\pgfdeclarelayer{nodelayer}
\pgfsetlayers{edgelayer,nodelayer,main}

\tikzset{initial text={}}

\tikzstyle{none}=[inner sep=0pt]
\tikzstyle{gate}=[rectangle,draw=black, inner sep=5pt,outer sep=1pt]
\tikzstyle{and}=[gate]
\tikzstyle{or}=[gate]
\tikzstyle{seq}=[gate]
\tikzstyle{neg}=[gate]
\tikzstyle{fail}=[gate]
\tikzstyle{cost}=[gate, minimum height=2.5em, rectangle split, rectangle split parts=2, rectangle split draw splits=false, inner sep=3pt]
\tikzstyle{trigger}=[gate]
\tikzstyle{costguard}=[above,font=\scriptsize]
\tikzstyle{smallgate}=[gate,inner sep=4pt]
\tikzstyle{urgent}=[font={\scriptsize{\sf \textcolor{Dimgray}{U}}}, inner sep=0pt]
\tikzstyle{state-circle}=[circle,draw=black,inner sep=2pt,minimum size=0.35cm,outer sep=2pt]
\tikzstyle{state}=[ellipse,draw=black,inner sep=2pt,outer sep=2pt]
\tikzstyle{state-box}=[rectangle,draw=black,inner sep=2pt,minimum size=0.35cm,outer sep=2pt]
\tikzstyle{probability-state}=[circle,fill=black,draw=black,inner sep=0pt,minimum size=1pt]
\tikzstyle{BETrig}=[regular polygon,regular polygon sides=3,draw=black, inner sep=0pt, outer sep=1pt,minimum width=0.85cm]
\tikzstyle{reset}=[rectangle,draw=black, inner sep=5pt]
\tikzstyle{nodename}=[text width=1.8cm, align=center]
\tikzstyle{bename}=[nodename, below=0.3]
\tikzstyle{gatename}=[nodename, above=0.3]
\tikzstyle{costname}=[nodename, above=0.4]
\tikzstyle{update}=[font=\scriptsize]

\tikzstyle{false}=[color={ACMRed!70!black}]

\tikzstyle{true}=[color={ACMGreen!70!black}]

\tikzstyle{undefined}=[color={black!70}]
\tikzstyle{tt}=[defender]
\tikzstyle{ff}=[attacker]
\tikzstyle{uu}=[neutral]

\tikzstyle{attacker}=[fill={ACMRed!30},postaction={pattern=horizontal lines,pattern color=ACMRed!60}]

\tikzstyle{defender}=[fill={ACMGreen!50},postaction={pattern=vertical lines,pattern color=ACMGreen!80}]

\tikzstyle{neutral}=[fill={gray!30},postaction={pattern=crosshatch dots,pattern color=gray!80},font=\bf]

\tikzstyle{nothing}=[]
\tikzstyle{arrow-urgent}=[-latex,draw=black]
\tikzstyle{arrow-non-urgent}=[-latex,draw=black,dashed]
\tikzstyle{arrow-simple}=[-latex,draw=black]
\tikzstyle{arrow-reset}=[-latex,dashed]
\tikzstyle{arrow-trigger}=[-latex,decoration={snake,segment length=4,amplitude=.6, post=lineto,post length=2pt},decorate]
\tikzstyle{arrow-segment-after-probabilities}=[-latex,draw=black]
\tikzstyle{markovian}=[->,draw=black,postaction={decorate},decoration={markings,mark=at position .5 with {\arrow{>}}}]
\tikzstyle{arrow-segment-before-probabilities-immediate}=[-,draw=black]
\tikzstyle{arrow-segment-before-probabilities-timed}=[-,draw=black]
\tikzstyle{stateText}=[rectangle,fill=White,draw=black, inner sep=5pt, minimum height=0pt]
\tikzstyle{non-urgent}=[densely dashed]

\newcommand*{\fancyrefapplabelprefix}{app}

\fancyrefaddcaptions{english}{%
  \providecommand*{\frefappname}{appendix}%
  \providecommand*{\Frefappname}{Appendix}%
}

\frefformat{plain}{\fancyrefapplabelprefix}{\frefappname\fancyrefdefaultspacing#1}
\Frefformat{plain}{\fancyrefapplabelprefix}{\Frefappname\fancyrefdefaultspacing#1}

\frefformat{vario}{\fancyrefapplabelprefix}{%
  \frefappname\fancyrefdefaultspacing#1#3%
}
\Frefformat{vario}{\fancyrefapplabelprefix}{%
  \Frefappname\fancyrefdefaultspacing#1#3%
}

\newcommand*{\fancyreflinelabelprefix}{line}

\fancyrefaddcaptions{english}{%
  \providecommand*{\freflinename}{line}%
  \providecommand*{\Freflinename}{Line}%
}

\frefformat{plain}{\fancyreflinelabelprefix}{\freflinename\fancyrefdefaultspacing#1}
\Frefformat{plain}{\fancyreflinelabelprefix}{\Freflinename\fancyrefdefaultspacing#1}

\frefformat{vario}{\fancyreflinelabelprefix}{%
  \freflinename\fancyrefdefaultspacing#1#3%
}
\Frefformat{vario}{\fancyreflinelabelprefix}{%
  \Freflinename\fancyrefdefaultspacing#1#3%
}

\newcommand*{\fancyrefexlabelprefix}{ex}

\fancyrefaddcaptions{english}{%
  \providecommand*{\frefexname}{example}%
  \providecommand*{\Frefexname}{Example}%
}

\frefformat{plain}{\fancyrefexlabelprefix}{\frefexname\fancyrefdefaultspacing#1}
\Frefformat{plain}{\fancyrefexlabelprefix}{\Frefexname\fancyrefdefaultspacing#1}

\frefformat{vario}{\fancyrefexlabelprefix}{%
  \frefexname\fancyrefdefaultspacing#1#3%
}
\Frefformat{vario}{\fancyrefexlabelprefix}{%
  \Frefexname\fancyrefdefaultspacing#1#3%
}

\newcommand*{\fancyrefdeflabelprefix}{def}

\fancyrefaddcaptions{english}{%
  \providecommand*{\frefdefname}{definition}%
  \providecommand*{\Frefdefname}{Definition}%
}

\frefformat{plain}{\fancyrefdeflabelprefix}{\frefdefname\fancyrefdefaultspacing#1}
\Frefformat{plain}{\fancyrefdeflabelprefix}{\Frefdefname\fancyrefdefaultspacing#1}

\frefformat{vario}{\fancyrefdeflabelprefix}{%
  \frefdefname\fancyrefdefaultspacing#1#3%
}
\Frefformat{vario}{\fancyrefdeflabelprefix}{%
  \Frefdefname\fancyrefdefaultspacing#1#3%
}

\newcommand*{\fancyreftheolabelprefix}{theo}

\fancyrefaddcaptions{english}{%
  \providecommand*{\freftheoname}{theorem}%
  \providecommand*{\Freftheoname}{Theorem}%
}

\frefformat{plain}{\fancyreftheolabelprefix}{\freftheoname\fancyrefdefaultspacing#1}
\Frefformat{plain}{\fancyreftheolabelprefix}{\Freftheoname\fancyrefdefaultspacing#1}

\frefformat{vario}{\fancyreftheolabelprefix}{%
  \freftheoname\fancyrefdefaultspacing#1#3%
}
\Frefformat{vario}{\fancyreftheolabelprefix}{%
  \Freftheoname\fancyrefdefaultspacing#1#3%
}

\newcommand*{\fancyrefalglabelprefix}{alg}

\fancyrefaddcaptions{english}{%
  \providecommand*{\frefalgname}{algorithm}%
  \providecommand*{\Frefalgname}{Algorithm}%
}

\frefformat{plain}{\fancyrefalglabelprefix}{\frefalgname\fancyrefdefaultspacing#1}
\Frefformat{plain}{\fancyrefalglabelprefix}{\Frefalgname\fancyrefdefaultspacing#1}

\frefformat{vario}{\fancyrefalglabelprefix}{%
  \frefalgname\fancyrefdefaultspacing#1#3%
}
\Frefformat{vario}{\fancyrefalglabelprefix}{%
  \Frefalgname\fancyrefdefaultspacing#1#3%
}

\newcommand*{\fancyreflemlabelprefix}{lem}

\fancyrefaddcaptions{english}{%
  \providecommand*{\freflemname}{lemma}%
  \providecommand*{\Freflemname}{Lemma}%
}

\frefformat{plain}{\fancyreflemlabelprefix}{\freflemname\fancyrefdefaultspacing#1}
\Frefformat{plain}{\fancyreflemlabelprefix}{\Freflemname\fancyrefdefaultspacing#1}

\frefformat{vario}{\fancyreflemlabelprefix}{%
  \freflemname\fancyrefdefaultspacing#1#3%
}
\Frefformat{vario}{\fancyreflemlabelprefix}{%
  \Freflemname\fancyrefdefaultspacing#1#3%
}

\newcommand*{\fancyrefcorlabelprefix}{cor}

\fancyrefaddcaptions{english}{%
  \providecommand*{\frefcorname}{corollary}%
  \providecommand*{\Frefcorname}{Corollary}%
}

\frefformat{plain}{\fancyrefcorlabelprefix}{\frefcorname\fancyrefdefaultspacing#1}
\Frefformat{plain}{\fancyrefcorlabelprefix}{\Frefcorname\fancyrefdefaultspacing#1}

\frefformat{vario}{\fancyrefcorlabelprefix}{%
  \frefcorname\fancyrefdefaultspacing#1#3%
}
\Frefformat{vario}{\fancyrefcorlabelprefix}{%
  \Frefcorname\fancyrefdefaultspacing#1#3%
}

\newcommand*{\fancyrefsubseclabelprefix}{subsec}

\fancyrefaddcaptions{english}{%
  \providecommand*{\frefsubsecname}{subsection}%
  \providecommand*{\Frefsubsecname}{Subsection}%
}

\frefformat{plain}{\fancyrefsubseclabelprefix}{\frefsubsecname\fancyrefdefaultspacing#1}
\Frefformat{plain}{\fancyrefsubseclabelprefix}{\Frefsubsecname\fancyrefdefaultspacing#1}

\frefformat{vario}{\fancyrefsubseclabelprefix}{%
  \frefsubsecname\fancyrefdefaultspacing#1#3%
}
\Frefformat{vario}{\fancyrefsubseclabelprefix}{%
  \Frefsubsecname\fancyrefdefaultspacing#1#3%
}

\newcommand*{\fancyrefreviewlabelprefix}{rev}

\fancyrefaddcaptions{english}{%
	\providecommand*{\frefreviewname}{review}%
	\providecommand*{\Frefreviewname}{Review}%
}

\frefformat{plain}{\fancyrefreviewlabelprefix}{\frefreviewname\fancyrefdefaultspacing#1}
\Frefformat{plain}{\fancyrefreviewlabelprefix}{\Frefreviewname\fancyrefdefaultspacing#1}

\frefformat{vario}{\fancyrefreviewlabelprefix}{%
	\frefreviewname\fancyrefdefaultspacing#1#3%
}
\Frefformat{vario}{\fancyrefreviewlabelprefix}{%
	\Frefreviewname\fancyrefdefaultspacing#1#3%
}

\renewcommand{\phi}{\upvarphi}
\renewcommand{\psi}{\uppsi}

\renewcommand{\rho}{\upvarrho}
\renewcommand{\upsilon}{\upupsilon}
\renewcommand{\chi}{\upchi}

\renewcommand{\max}{\mathop{\mathsf{max}}}
\renewcommand{\min}{\mathop{\mathsf{min}}}

 \newcommand{\formattool}[1]{\textsf{#1}}
\newcommand{\formatstandard}[1]{\textsf{#1}}
\newcommand{\formatnames}[1]{\textsl{#1}}

\newcommand{\adt}{\formatnames{ADT}}
\newcommand{\PAC}{\formatnames{PAC}}
\newcommand{\DOT}{\formatstandard{DOT}}
\newcommand{\XML}{\formatstandard{XML}}

\newcommand{\UPPAAL}{\formattool{UPPAAL}}
\newcommand{\prism}{\formattool{PRISM}}
\newcommand{\quadtool}{\formattool{QuADTool}}

\newcommand{\modest}{\formattool{MODEST}}
\newcommand{\storm}{\formattool{STORM}}
\newcommand{\prismgames}{\formattool{PRISM-games}}
\newcommand{\jani}{\formatstandard{JANI}}
\newcommand{\adtool}{\formattool{ADTool}}

\newcommand{\atbest}{\formatstandard{ATBEST}}
\newcommand{\atbestURL}{\url{https://www.model.in.tum.de/~kraemerj/upload/}}

\newcommand{\risqflan}{\formattool{RisQFLan}}
\newcommand{\attop}{\formattool{ATTop}}

\newcommand{\MultiVeStA}{\formattool{MultiVeStA}}

\newcommand{\orcidID}[1]{{\href{https://orcid.org/#1}{\protect\raisebox{3.25pt}{\protect\includegraphics{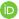}}}}}

\newcommand{\julia}[1]{}
\newcommand{\jan}[1]{}
\newcommand{\katharina}[1]{}

\newenvironment{lemma*}[1]{\medskip\noindent\textbf{Recap of Lemma\ifthenelse{\isempty{#1}}{}{~#1}.}}{}
\newenvironment{theorem*}[1]{\medskip\noindent\textbf{Recap of Theorem\ifthenelse{\isempty{#1}}{}{~#1}.}}{}

\newcommand{\true}{\mathsf{tt}}
\newcommand{\false}{\mathsf{ff}}

\DeclareDocumentCommand{\variable}{O{} D<>{}}{\mathsf{v}_{#1}^{#2}}
\DeclareDocumentCommand{\variables}{O{} D<>{}}{\mathcal{V}_{#1}^{#2}}

\newcommand{\set}[1]{\{#1\}}
\newcommand{\powerset}[1]{2^{#1}}
\newcommand{\disjointUnion}{\mathop{\dot{\cup}}}

\newcommand{\RR}{\mathbb{R}}

\DeclareDocumentCommand{\anySet}{O{} D<>{}}{\mathcal{X}_{#1}^{#2}}

\DeclareDocumentCommand{\alg}{O{} D<>{} D(){}}{\mathcal{A}_{#1}^{#2}\ifthenelse{\isempty{#3}}{}{(#3)}}

\DeclareDocumentCommand{\distribution}{O{} D<>{} D(){}}{\mu_{#1}^{#2}\ifthenelse{\isempty{#3}}{}{(#3)}}

\DeclareDocumentCommand{\probability}{O{} D<>{} D(){}}{\mathbb{P}_{#1}^{#2}\ifthenelse{\isempty{#3}}{}{\big(#3\big)}}
\DeclareDocumentCommand{\expectation}{O{} D<>{} D(){}}{\mathbb{E}_{#1}^{#2}\ifthenelse{\isempty{#3}}{}{[#3]}}
\DeclareDocumentCommand{\support}{O{} D<>{} D(){}}{\mathsf{Supp}_{#1}^{#2}\ifthenelse{\isempty{#3}}{}{(#3)}}

\DeclareDocumentCommand{\cylinder}{O{} D<>{} D(){}}{\mathsf{Cyl}_{#1}^{#2}\ifthenelse{\isempty{#3}}{}{\big(#3\big)}}
\DeclareDocumentCommand{\algebra}{O{} D<>{} D(){}}{\mathcal{F}_{#1}^{#2}\ifthenelse{\isempty{#3}}{}{\big(#3\big)}}
\DeclareDocumentCommand{\probspace}{O{} D<>{}}{\big(\Omega_{#1}^{#2},\mathcal{F}_{#1}^{#2},\probability\big)}
\DeclareDocumentCommand{\density}{O{} D<>{} D(){}}{\mathsf{f}_{#1}^{#2}\ifthenelse{\isempty{#3}}{}{(#3)}}

\DeclareDocumentCommand{\randomVariable}{O{} D<>{} D(){} D(){}}{\mathsf{C}_{#1}^{#2}\ifthenelse{\isempty{#3}}{}{(#3)}\ifthenelse{\isempty{#4}}{}{(#4)}}

\newcommand{\opformat}[1]{\mathbin{\mathtt{#1}}}
\newcommand{\opAnd}{\opformat{AND}}
\newcommand{\opOr}{\opformat{OR}}
\newcommand{\opSand}{\opformat{SAND}}
\newcommand{\opSor}{\opformat{SOR}}

\newcommand{\opReset}{\opformat{RE}}
\newcommand{\opTrigger}{\opformat{TR}}
\newcommand{\opNeg}{\opformat{NOT}}

\DeclareDocumentCommand{\opVot}{O{k/n}}{\opformat{VOT}(#1)}

\DeclareDocumentCommand{\graph}{O{} D<>{}}{\mathsf{G}_{#1}^{#2}}

\DeclareDocumentCommand{\vertex}{O{} D<>{}}{\mathsf{v}_{#1}^{#2}}
\DeclareDocumentCommand{\othervertex}{O{} D<>{}}{\mathsf{u}_{#1}^{#2}}
\DeclareDocumentCommand{\vertices}{O{} D<>{}}{\mathsf{V}_{#1}^{#2}}

\DeclareDocumentCommand{\edge}{O{} D<>{}}{\mathsf{e}_{#1}^{#2}}
\DeclareDocumentCommand{\edges}{O{} D<>{}}{\mathsf{E}_{#1}^{#2}}

\DeclareDocumentCommand{\inp}{O{} D<>{} D(){}}{\mathsf{in}_{#1}^{#2}\ifthenelse{\isempty{#3}}{}{(#3)}}

\DeclareDocumentCommand{\exampleGraph}{O{} D<>{}}{\graph[#1]<#2>=(\vertices[#1]<#2>, \edges[#1]<#2>)}

\DeclareDocumentCommand{\AFT}{O{} D<>{}}{\mathsf{AFT}_{#1}^{#2}}
\DeclareDocumentCommand{\ADT}{O{} D<>{}}{\mathsf{ADT}_{#1}^{#2}}
\DeclareDocumentCommand{\ADD}{O{} D<>{}}{\mathsf{ADT}_{#1}^{#2}}
\DeclareDocumentCommand{\AT}{O{} D<>{}}{\mathsf{AT}_{#1}^{#2}}

\DeclareDocumentCommand{\topGoal}{O{} D<>{}}{\mathsf{g}_{#1}^{#2}}

\DeclareDocumentCommand{\operator}{O{} D<>{} D(){}}{\mathsf{op}_{#1}^{#2}\ifthenelse{\isempty{#3}}{}{(#3)}}
\DeclareDocumentCommand{\operators}{O{} D<>{} D(){}}{\mathsf{O}_{#1}^{#2}\ifthenelse{\isempty{#3}}{}{(#3)}}

\DeclareDocumentCommand{\basicEvent}{O{} D<>{}}{\mathsf{be}_{#1}^{#2}}
\DeclareDocumentCommand{\basicEvents}{O{} D<>{}}{\mathsf{BE}_{#1}^{#2}}
\DeclareDocumentCommand{\basicComponent}{O{} D<>{}}{\mathsf{bc}_{#1}^{#2}}
\DeclareDocumentCommand{\basicComponents}{O{} D<>{}}{\mathsf{BC}_{#1}^{#2}}
\DeclareDocumentCommand{\basicAttack}{O{} D<>{}}{\mathsf{ba}_{#1}^{#2}}
\DeclareDocumentCommand{\basicAttacks}{O{} D<>{}}{\mathsf{BA}_{#1}^{#2}}
\DeclareDocumentCommand{\basicDefense}{O{} D<>{}}{\mathsf{bd}_{#1}^{#2}}
\DeclareDocumentCommand{\basicDefenses}{O{} D<>{}}{\mathsf{BD}_{#1}^{#2}}

\DeclareDocumentCommand{\probabilityFunction}{O{} D<>{} D(){}}{\mathsf{Pr}_{#1}^{#2}\ifthenelse{\isempty{#3}}{}{(#3)}}
\DeclareDocumentCommand{\delayFunction}{O{} D<>{}  D(){}}{\mathsf{Delay}^{#2}\ifthenelse{\isempty{#3}}{}{(#3)}}
\DeclareDocumentCommand{\costFunction}{O{} D<>{}  D(){}}{\mathsf{Cost}_{E}^{#2}\ifthenelse{\isempty{#3}}{}{(#3)}}
\DeclareDocumentCommand{\distributionFunction}{O{} D<>{} D(){}}{\mathsf{Dist}^{#2}\ifthenelse{\isempty{#3}}{}{(#3)}}
\DeclareDocumentCommand{\costDelayFunction}{O{} D<>{} D(){}}{\mathsf{Cost}_{D}^{#2}\ifthenelse{\isempty{#3}}{}{(#3)}}

\newcommand{\gates}{\mathsf{G}}
\newcommand{\gate}{\mathsf{g}}
\newcommand{\subscript}[2]{$#1 _ #2$}
 
\DeclareDocumentCommand{\triggerTransitions}{O{} D<>{}}{\mathsf{TEdge}_{#1}^{#2}}
\DeclareDocumentCommand{\triggerBE}{O{} D<>{}}{\basicEvents[\opTrigger\ifthenelse{\isempty{#1}}{}{,#1}]<#2>}

\DeclareDocumentCommand{\resetTransitions}{O{} D<>{}}{\mathsf{REdge}_{#1}^{#2}}
\DeclareDocumentCommand{\resetBE}{O{} D<>{}}{\basicEvents[\opReset\ifthenelse{\isempty{#1}}{}{,#1}]<#2>}

\DeclareDocumentCommand{\typeFunction}{O{} D<>{} D(){}}{\mathsf{t}_{#1}^{#2}\ifthenelse{\isempty{#3}}{}{(#3)}}

\DeclareDocumentCommand{\exampleAFT}{O{} D<>{}}{\AFT[#1]<#2>=(\vertices[#1]<#2>,
  \edges[#1]<#2>, \typeFunction[#1]<#2>,
  \triggerTransitions[#1]<#2>)}
\DeclareDocumentCommand{\exampleADD}{O{} D<>{}}{\ADT[#1]<#2>=(\vertices[#1]<#2>,
  \edges[#1]<#2>, \typeFunction[#1]<#2>)} 
\DeclareDocumentCommand{\exampleAT}{O{} D<>{}}{\AT[#1]<#2>=(\vertices[#1]<#2>,
  \edges[#1]<#2>, \typeFunction[#1]<#2>)} 

\DeclareDocumentCommand{\labellingFunction}{O{} D<>{} D(){}}{\ell_{#1}^{#2}\ifthenelse{\isempty{#3}}{}{(#3)}}

\DeclareDocumentCommand{\type}{O{} D<>{} D(){}}{\mathsf{T}_{#1}^{#2}\ifthenelse{\isempty{#3}}{}{(#3)}}
\DeclareDocumentCommand{\randomVariable}{O{} D<>{} D(){}}{\mathsf{X}_{#1}^{#2}\ifthenelse{\isempty{#3}}{}{(#3)}}
\DeclareDocumentCommand{\val}{O{} D<>{} D(){}}{\mathsf{v}_{#1}^{#2}\ifthenelse{\isempty{#3}}{}{(#3)}}
\DeclareDocumentCommand{\metric}{O{} D<>{} D(){}}{\mathsf{d}_{#1}^{#2}\ifthenelse{\isempty{#3}}{}{(#3)}}
\DeclareDocumentCommand{\error}{O{} D<>{} D(){}}{\epsilon_{#1}^{#2}\ifthenelse{\isempty{#3}}{}{(#3)}}
\DeclareDocumentCommand{\prob}{O{} D<>{} D(){}}{\delta_{#1}^{#2}\ifthenelse{\isempty{#3}}{}{(#3)}}

\DeclareDocumentCommand{\interpretation}{O{} D<>{} D(){}D(){}}{\mathsf{I}_{#1}^{#2}\ifthenelse{\isempty{#3}}{}{(#3)}\ifthenelse{\isempty{#4}}{}{(#4)}}
\DeclareDocumentCommand{\pac}{O{} D<>{} D||{} D(){}}{\mathsf{pac}_{#1}^{#2}\ifthenelse{\isempty{#3}}{}{(#3)}\ifthenelse{\isempty{#4}}{}{(#4)}}
\DeclareDocumentCommand{\values}{O{} D<>{} D(){}}{\mathsf{val}_{#1}^{#2}\ifthenelse{\isempty{#3}}{}{(#3)}}
\DeclareDocumentCommand{\valuesBE}{O{\basicEvents} D<>{} D(){}}{\mathsf{val}_{#1}^{#2}\ifthenelse{\isempty{#3}}{}{(#3)}}

\DeclareDocumentCommand{\vector}{O{} D<>{}  D(){}}{\avector{\alpha}_{#1}^{#2}\ifthenelse{\isempty{#3}}{}{(#3)}}
\DeclareDocumentCommand{\matrix}{O{} D<>{}  D(){}}{\ifstrequal{#1}{1}{\amatrix{A}}{\ifstrequal{#1}{2}{\amatrix{B}}{\amatrix{A}}}_{\ifstrequal{#1}{1}{}{\ifstrequal{#1}{2}{}{#1}}}^{#2}\ifthenelse{\isempty{#3}}{}{(#3)}}
\DeclareDocumentCommand{\aph}{O{} D<>{} D(){}}{\ensuremath{\mathsf{APH}_{#1}^{#2}\ifthenelse{\isempty{#3}}{}{(#3)}}}
\DeclareDocumentCommand{\setAPH}{O{} D<>{} D(){}}{\mathcal{APH}_{#1}^{#2}\ifthenelse{\isempty{#3}}{}{(#3)}}
\DeclareDocumentCommand{\exampleAPH}{O{} D<>{} D(){}}{\aph[#1]<#2> = (\vector[#1]<#2>,\matrix[#1]<#2>)\ifthenelse{\isempty{#3}}{}{(#3)}}
\DeclareDocumentCommand{\norm}{m O{} D<>{}}{\left\lVert#1\right\rVert_{#2}^{#3}}

\DeclareDocumentCommand{\probabilityFunction}{O{} D<>{} D(){}}{\mathsf{Pr}_{#1}^{#2}\ifthenelse{\isempty{#3}}{}{(#3)}}
\DeclareDocumentCommand{\costFunction}{O{} D<>{}  D(){}}{\mathsf{Cost}_{E}^{#2}\ifthenelse{\isempty{#3}}{}{(#3)}}
\DeclareDocumentCommand{\distributionFunction}{O{} D<>{} D(){}}{\mathsf{Dist}^{#2}\ifthenelse{\isempty{#3}}{}{(#3)}}
\DeclareDocumentCommand{\costDelayFunction}{O{} D<>{} D(){}}{\mathsf{Cost}_{D}^{#2}\ifthenelse{\isempty{#3}}{}{(#3)}}

\DeclareDocumentCommand{\cost}{O{} D<>{} D(){}}{\mathsf{C}_{#1}^{#2}\ifthenelse{\isempty{#3}}{}{(#3)}}
\DeclareDocumentCommand{\costVars}{O{} D<>{} D(){}}{\mathsf{CVars}_{#1}^{#2}\ifthenelse{\isempty{#3}}{}{(#3)}}
\DeclareDocumentCommand{\bounds}{O{} D<>{} D(){}}{\mathcal{B}_{#1}^{#2}\ifthenelse{\isempty{#3}}{}{(#3)}}
\DeclareDocumentCommand{\operatorCost}{O{} D<>{} D(){}}{\mathop{\mathsf{op}}_{#1}^{#2}\ifthenelse{\isempty{#3}}{}{(#3)}}

\DeclareDocumentCommand{\totNumAttacks}{O{} D<>{}  D(){}}{\mathsf{A}_{#1}^{#2}\ifthenelse{\isempty{#3}}{}{(#3)}}
\DeclareDocumentCommand{\totNumBlocks}{O{} D<>{} D(){}}{\mathsf{B}_{#1}^{#2}\ifthenelse{\isempty{#3}}{}{(#3)}}
\DeclareDocumentCommand{\frequency}{O{} D<>{} D(){}}{\mathsf{F}_{#1}^{#2}\ifthenelse{\isempty{#3}}{}{(#3)}}

\DeclareDocumentCommand{\successProbability}{O{} D<>{} D(){}}{\mathsf{V}_{#1}^{#2}\ifthenelse{\isempty{#3}}{}{(#3)}}

\newcolumntype{L}[1]{>{\hsize=#1\hsize\raggedright\arraybackslash}X}%
\newcolumntype{R}[1]{>{\hsize=#1\hsize\raggedleft\arraybackslash}X}%
\newcolumntype{C}[1]{>{\hsize=#1\hsize\centering\arraybackslash}X}%

\newcolumntype{R}[2]{%
    >{\adjustbox{angle=#1,lap=\width-(#2)}\bgroup}%
    l%
    <{\egroup}%
}

\DeclareDocumentCommand{\bijection}{O{} D<>{} D(){}}{\mathsf{bi}_{#1}^{#2}\ifthenelse{\isempty{#3}}{}{(#3)}}

\newcommand{\Id}{\ensuremath{\mathop{\mathsf{Id}}}}

{%
\begingroup 
\setlength\arraycolsep{1ex}
\begin{bmatrix}
}%
{%
\end{bmatrix}
\endgroup 
}%

\newcommand{\amatrix}[1]{\mathbf{#1}}
\newcommand{\avector}[1]{\boldsymbol{#1}}

\DeclareDocumentCommand{\trace}{O{} D<>{} D(){}}{\mathsf{tr}_{#1}^{#2}\ifthenelse{\isempty{#3}}{}{(#3)}}
\DeclareDocumentCommand{\traces}{O{} D<>{} D(){}}{\mathsf{Tr}_{#1}^{#2}\ifthenelse{\isempty{#3}}{}{(#3)}}
\DeclareDocumentCommand{\execution}{O{} D<>{} D(){}}{\mathsf{ex}_{#1}^{#2}\ifthenelse{\isempty{#3}}{}{(#3)}}
\DeclareDocumentCommand{\executions}{O{} D<>{} D(){}}{\mathsf{Ex}_{#1}^{#2}\ifthenelse{\isempty{#3}}{}{(#3)}}
\DeclareDocumentCommand{\valuation}{O{} D<>{} D(){}}{\mathsf{eval}_{#1}^{#2}\ifthenelse{\isempty{#3}}{}{(#3)}}

\DeclareDocumentCommand{\sensitivity}{O{} D<>{} D(){}}{\mathsf{Sens}_{#1}^{#2}\ifthenelse{\isempty{#3}}{}{(#3)}}
\DeclareDocumentCommand{\specificity}{O{} D<>{}  D(){}}{\mathsf{Spec}_{#1}^{#2}\ifthenelse{\isempty{#3}}{}{(#3)}}
\DeclareDocumentCommand{\ppv}{O{} D<>{}  D(){}}{\mathsf{PPV}_{#1}^{#2}\ifthenelse{\isempty{#3}}{}{(#3)}}
\DeclareDocumentCommand{\weight}{O{} D<>{}  D(){}}{\mathsf{w}_{#1}^{#2}\ifthenelse{\isempty{#3}}{}{(#3)}}
\DeclareDocumentCommand{\quantile}{O{} D<>{}  D(){}}{\mathsf{q}_{#1}^{#2}\ifthenelse{\isempty{#3}}{}{(#3)}}
\DeclareDocumentCommand{\fitness}{O{} D<>{}  D(){}}{\mathsf{F}_{#1}^{#2}\ifthenelse{\isempty{#3}}{}{(#3)}}
\DeclareDocumentCommand{\generations}{O{} D<>{}  D(){}}{\mathsf{Gen}_{#1}^{#2}\ifthenelse{\isempty{#3}}{}{(#3)}}
\DeclareDocumentCommand{\generationwidth}{O{} D<>{}  D(){}}{\mathsf{GenW}_{#1}^{#2}\ifthenelse{\isempty{#3}}{}{(#3)}}
\DeclareDocumentCommand{\mutationrate}{O{} D<>{}  D(){}}{\mathsf{mutr}_{#1}^{#2}\ifthenelse{\isempty{#3}}{}{(#3)}}
\DeclareDocumentCommand{\trainingrate}{O{} D<>{}  D(){}}{\mathsf{trainr}_{#1}^{#2}\ifthenelse{\isempty{#3}}{}{(#3)}}
\DeclareDocumentCommand{\actions}{O{} D<>{} D(){}}{\mathsf{Act}_{#1}^{#2}\ifthenelse{\isempty{#3}}{}{(#3)}}
\DeclareDocumentCommand{\variableAlg}{m O{} D<>{}  D(){}}{\mathsf{#1}_{#2}^{#3}\ifthenelse{\isempty{#4}}{}{(#4)}}

\DeclareDocumentCommand{\operator}{O{} D<>{} D(){}}{\mathsf{op}_{#1}^{#2}\ifthenelse{\isempty{#3}}{}{(#3)}}

\DeclareDocumentCommand{\labelfunction}{O{} D<>{} D(){}}{\mathsf{c}_{#1}^{#2}\ifthenelse{\isempty{#3}}{}{(#3)}}

\DeclareDocumentCommand{\signatureFunction}{O{} D<>{} D(){}}{\mathsf{s}_{#1}^{#2}\ifthenelse{\isempty{#3}}{}{(#3)}}

\DeclareDocumentCommand{\randmod}{O{} D<>{} D(){}}{\mathsf{\epsilon}_{#1}^{#2}\ifthenelse{\isempty{#3}}{}{(#3)}}

\DeclareDocumentCommand{\mean}{O{} D<>{} D(){}}{\mathsf{\mu}_{#1}^{#2}\ifthenelse{\isempty{#3}}{}{(#3)}}
\DeclareDocumentCommand{\sd}{O{} D<>{} D(){}}{\mathsf{\sigma}_{#1}^{#2}\ifthenelse{\isempty{#3}}{}{(#3)}}

\newcounter{MYtempeqncnt}

\title{QuADTool: Attack-Defense-Tree Synthesis, Analysis and Bridge to Verification
}

\author{Florian Dorfhuber\inst{1,2} \Letter \email{florian.dorfhuber@in.tum.de}
	\orcidID{0000-0003-3755-7171}
	 \and \\ Julia Eisentraut\inst{2,3} \email{julia.eisentraut@posteo.de} \orcidID{0000-0002-7735-8751}
	 \and \\ Katharina Klioba\inst{4} \email{katharina.klioba@tuhh.de} \orcidID{0009-0002-7946-917X}
	 \and \\ Jan K\v{r}et\'{i}nsk\'{y}\inst{1,2} \email{jan.kretinsky@fi.muni.cz}
	 \orcidID{0000-0002-8122-2881}
	 }

\institute{Masaryk University, Brno, Czech Republic \and Technical University of Munich, Munich, Germany
	\and State Parliament of Northrhine-Westfalia, Germany \and Technical University of Hamburg, Hamburg, Germany}

\newcommand{\para}[1]{\medskip\noindent\textbf{#1}}

\begin{document}
\maketitle

\pagestyle{plain}

\begin{abstract}
  %
  Ranking risks and countermeasures is one of the foremost goals of
  quantitative security analysis. One of the popular frameworks, used also in industrial practice, for this task are \emph{attack-defense trees}.
  Standard quantitative analyses
  available for attack-defense trees can distinguish likely from
  unlikely vulnerabilities.
  We provide a tool that allows for easy synthesis and analysis of those models, also featuring probabilities, costs and time. Furthermore, it provides a variety of interfaces to existing model checkers and analysis tools.

  Unfortunately, currently available tools rely on \emph{precise}
  quantitative inputs (probabilities, timing, or costs of attacks),
  which are rarely available.
  Instead, only statistical, imprecise information is typically available, leaving us with \emph{probably approximately correct} (\PAC{})
  estimates of the real quantities.
  As a part of our tool, we extend the standard analysis techniques so they can handle the PAC input and yield rigorous bounds on the imprecision and uncertainty of the final result of the  analysis. 
  
  
  %
  
\end{abstract}


\section{Introduction}
\label{sec:introduction}
\emph{Attack trees}, e.g., ~\cite{kordy2013,widel2019}, and their extensions are a widespread formalism for threat modeling and risk assessment. Their applications range from analyzing attacks on smart grids~\cite{beckers2014}, ATMs~\cite{fraile2016}, optical power meters~\cite{fila2019}, SCADA control systems~\cite{byres2004,mcqueen2006,ten2007,lopez2012,al2019} or software supply chains~\cite{ohm2020} to intelligent autonomous vehicles and vehicular networks~\cite{houmer2020,kim2020,roblesmodel}, secure deployment of HTTPS~\cite{soligo2020} or cybersecurity awareness trainings for election officials~\cite{schurmann2020}. Moreover, attack trees are particularly useful for meta-modeling threats, combining several risk analyses together, e.g., in telemedicine~\cite{kim2020b,ramos2020}, impersonation attacks in e-learning~\cite{rosmansyah2020} or IoT systems~\cite{krichen2019}.  Besides these individual usages, attack trees are generally recommended as a means to identify the point of the weakest resistance according to the \emph{OWASP CISO AppSec Guide},\footnote{\url{https://www.owasp.org/index.php/CISO\_AppSec\_Guide:\_Criteria\_for\_Managing\_Application\_Security\_Risks}} or to retrospectively understand attacks (through so-called red teams) by NATO's \emph{Improving Common Security Risk Analysis} report~\cite{nato2008}.

In practice, there are not enough resources to fix every single vulnerability found in security assessment. Hence, the need to rank their importance arises. \emph{Quantitative analysis} of attack trees can utilize quantitative information such as probabilities, costs, and timing to compute how likely particular vulnerabilities are, how much the defense may cost, or how long attacks might take.  Consequently, it can distinguish unrealistically costly or improbable attacks from those that need to be fixed.  While there is a recent body of theoretical work in this direction, e.g.~\cite{kumar2015,hermanns2016,aslanyan2016,pekergin2016,kumar2017,kordy2018,eisentraut2019}, there are two main hindrances to practical use.  Firstly, it is the very \textit{limited tool support} for convenient quantitative modeling and analysis. Secondly, it is the \textit{uncertainty and imprecision of the quantitative input} information, arising in this domain even more than elsewhere. Indeed, apart from standard questions like \emph{What is the actual failure rate of the security camera?}, we may need to ask \emph{What is the probability of guessing a password using a database of the most common passwords and one year of CPU time?} or \emph{How much faster and cheaper is a particular attack going to be next year based on past years' data?} 

We address both issues: we provide a convenient \emph{tool}, which also features quantitative analysis capable of rigorously reflecting the \emph{uncertainty and imprecision of the quantitative inputs}.  To this end, we formalize such models as \emph{\PAC{}-quantitative trees}. Then, we show how to obtain them practically using our tool and provide an algorithm for their analysis.
In order to properly evaluate our tool and the \PAC{}-input quantitative analysis we create a benchmark suite~\atbest{} of $\ADT$ found in literature and ones generated randomly. 
We make all models available to the public.\footnote{For review, all models are in our artifact at \url{https://zenodo.org/records/11090469}.\label{fn:artifact}}

\para{Tool.}  We provide~\quadtool{}, a tool for the whole modeling-analysis workflow. The tool can import trees in the \DOT{} format and \XML{} format of \adtool{}~\cite{gadyatskaya2016}. The graphical interface allows for convenient creating, editing, and combining trees and generating \PAC{}-quantities directly from provided data. Further, the tool translates the trees to formats of model checkers \UPPAAL{}~\cite{behrmann2006}, \prism{}~\cite{kwiatkowska2011}, \prismgames{}~\cite{chen2013}, \modest{}~\cite{hartmanns2014}, and \storm{}~\cite{dehnert2017} (via \jani{}~\cite{budde2017} export of \modest{}). Consequently, the GUI effectively interfaces with other available analyses (such as~\cite{gadyatskaya2016,ruijters2017}) and to our new \PAC{}-input analysis. Besides, to help non-experts successfully construct meaningful models, we provide detailed automatically generated feedback on compliance with the assumptions of each of the offered analyses.

\para{\PAC{}-quantitative attack-defense trees.}  All the previously cited quantitative analyses assume that the quantities (transition probabilities, costs, delays) are given as exact numbers, allowing analyses to output exact numbers, too.  However, in reality, we often lack perfect quantities. 
Nonetheless, this data allows us to derive estimates of the actual value, possibly together with confidence intervals bounding these estimates' uncertainty. Hence, we assume the input quantities are \emph{probably approximately correct (\PAC{})}, i.e., with a given probability in a given interval around the estimate. We demonstrate how easily such trees can be obtained, relying on \emph{statistics} or \emph{time-series analysis}. To that end, we consider independent samples of the unknown quantity data (e. g., past years' information). We then extend the traditional bottom-up analyses to bound the final \PAC{} uncertainty resulting from the input quantities' \PAC{} uncertainty. Interestingly, our extended analysis comes with little additional computational cost in comparison to the standard analysis on average-size attack-defense trees.  Consequently,  our approach allows practitioners to stick to their current analysis and learn about the quality and reliability of its results with almost no additional effort.

\medskip

To summarize, our contribution is threefold:
\begin{itemize}[noitemsep,nosep]
\item Our tool~\quadtool{} addresses practical needs in attack-defense tree modeling and analysis, including a GUI and a CLI for editing and analyzing \adt{}. It also gives feedback on suitable exports for further analyses. 
\item The tool features a novel quantitative analysis of \PAC{}-attack-defense trees, the first extensions of attack-defense trees reflecting the data inaccuracy.
\item The tool comes equipped with \emph{AT BEnchmark SuiTe} (\atbest{}), collecting benchmarks from literature as well as randomly generated ones.
\end{itemize}

\para{Organization.} The paper is organized as follows. After exemplifying the framework of attack-defense trees in \Fref{sec:attack-defense-trees}, the tool functionality is presented in
\Fref{sec:functionality}. We describe our \PAC{}-input analysis in \Fref{sec:pac}.
 \Fref{sec:experiments} evaluates our contributions experimentally and \Fref{sec:conclusion} concludes. We provide an artifact with the tool and models\footref{fn:artifact}.

\medskip

\para{Related work.}
Besides two industrial tools dedicated to attack-defense tree analysis\footnote{see \url{http://www.amenaza.com} and \url{http://www.isograph.com/software/}}, which neither support more complex attack-defense tree structures such as~\cite{hermanns2016}, nor export to standard model checking tools, there are three dedicated (and up-to-date) \emph{quantitative} attack-defense tree analysis tools:
\begin{itemize}[nosep]
\item \adtool{}~\cite{gadyatskaya2016} supports attack-defense trees with the operators~$\opAnd$, $\opOr$ and Sequential-$opAnd$~$\opSand$ and attack-defense trees with operators~$\opAnd$ and $\opOr$. It uses the standard bottom-up traversal for quantitative model checking. In contrast to our tool, \adtool{} does not support \PAC{}-input analysis.   Additionally, one cannot add \PAC{}-input analysis to \adtool{} as a `custom' domain since multi-parameter attack-defense trees are not supported. We compare the performance of the bottom-up analysis using \adtool{} and \quadtool{} in \Fref{sec:experiments}.
\item \attop{}~\cite{ruijters2017} is a software bridging tool, which provides various horizontal and vertical model transformations to the \adtool{}~2.0 input format or to the model checker UPPAAL, however not to \prism{} or \modest{}. \PAC{} models and \PAC{} analysis are not supported.
\item \risqflan{}\cite{ter2021quantitative} allows specifying attack-defense trees with quantitative constraints using a dedicated probabilistic language (XTEXT grammar)\footnote{While the XTEXT files get compiled to dot-files for graphical display, these are currently not entirely compatible with the \quadtool{} import.}. Statistical model checking and precise analysis using \storm{} and \prism{} are supported. However, \risqflan{} neither supports \PAC{}-inputs nor continuous-time. Additionally, exports to \storm{} and \prism{} rely on discrete-time Markov chains rather than games or stochastic timed automata. Stochastic model checking in \risqflan{} relies on \MultiVeStA{}\cite{sebastio2013multivesta}. Hence \risqflan{} does not allow synthesizing optimal attack and defense strategies. However, \risqflan{} allows assigning detection-rates to inner nodes and to separate defenses from countermeasures (which we do not feature). Hence, a complete performance comparision between \risqflan{} and \quadtool{} is not possible.
\end{itemize}

All these quantitative analyses have precise inputs and outputs (rather than \PAC{} values). The resulting uncertainty is not bound, i.e., there is no estimate of how likely the analysis results are close to the actual results. In contrast our \PAC{} analysis takes uncertainty of each input value into account. Hence, our resulting confidence interval covers likely deviations and bounds the uncertainty, which previous quantitative analysis cannot.

  Surprisingly, the uncertainty of the input quantities in attack-defense trees has not gained too much attention yet~\cite{buoni2010,mezei2011,pekergin2016}. 
  Buldas et al.~\cite{buldas2020} present the first approach to account for inconsistency and gaps in historical data and experts' estimates in attack-defense trees. The main limitation of their approach is the \emph{impossibility of reflecting the data accuracy, which is the main contribution of this paper}.
  
  In \cite{lopuhaa2022efficient} a new approach to extend quantitative analysis on Attack Trees is presented, that may lift the structural bounds our tool currently has, but does cover a smaller set of operators compared to our approach. Additionally, there is no support for defense operations, yet.


\section{Attack-Defense Trees}
\label{sec:attack-defense-trees}
In this section, we briefly recall the notion of \emph{attack-defense trees} (\adt), e.g.,~\cite{kordy2011,hermanns2016}. Attack-defense trees are labeled trees, where the root represents the overall goal to attack a system. This goal can be recursively refined into smaller sub-goals, giving rise to descendant nodes, until no further refinement is possible. Inner nodes receive operators as labels binding the roles of the children. The leaves represent \emph{basic events} (denoted $\basicEvents$), unique observable happenings in
the real world. They are partitioned into \emph{basic attack steps} and \emph{basic defense steps}.

\begin{definition}[Attack-Defense Trees~\cite{hermanns2016}]
	\label{def:add}
	An \emph{attack-defense tree}\footnote{In literature, \adt{} are often directed acyclic graphs (DAG), yet called trees. Hence, we call our framework attack-defense trees, although the underlying structure is a graph.} (over a set of basic events~$\basicEvents$) is a tuple $\exampleADD$ where
	\begin{itemize}
		\item $\left (\vertices, \edges \right)$ is a \emph{directed acyclic graph}, with a designated \emph{goal} sink vertex~$\topGoal$ for the attacker. The source vertices $\basicEvents \subseteq \vertices$ are called \emph{basic events}. All other vertices, \mbox{$\gates \coloneqq \vertices \setminus \basicEvents$}, are called \emph{gates}; direct predecessors w.r.t. $(\vertices,\edges)$ of each gate are called its \emph{inputs}.
		\item We additionally require that gates labeled with~$\opAnd$, $\opOr$ and their sequential versions $\opSand$ or $\opSor$  have at least two inputs and that other gates have exactly one input.
		\item $\typeFunction{} \colon\ \gates \to \operators$ is the \emph{type function} assigning an operator to each gate. 
		\item $\triggerTransitions \subseteq \set{\vertex \in \gates \mid \typeFunction(\vertex) = \opTrigger} \times \basicEvents$ are \emph{trigger} edges from $\opTrigger$ gates to BE;       
		\item
		$\resetTransitions \subseteq \set{\vertex \in \gates \mid \typeFunction(\vertex) = \opReset}\times \basicEvents$ are \emph{reset} edges from $\opReset$ gates to BE;
	\end{itemize}
\end{definition}

\begin{figure*}[htbp]
	\centering
	\includegraphics[width=\textwidth]{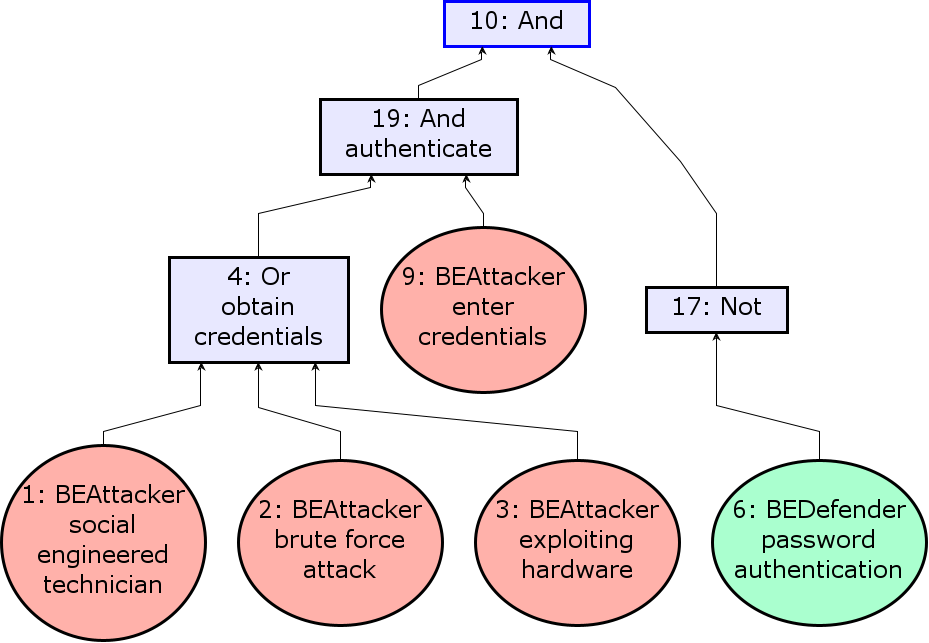}
	\caption{The $\ADT{}$ from \cite{fila2019} represents ways to gain access to a power meter. We will use this model also in \Fref{sec:pacanalysis}. In \Fref{sec:experiments} we evaluate our tool on more complex attack-defense trees. }
	\label{fig:toyexample}
\end{figure*}


\begin{example}
	\label{ex:toyexample}
	In \Fref{fig:toyexample}, we present a small attack-defense tree representing an attacker (nodes in red) who wants to maliciously gain access to a digital power meter. The goal node at the top (with a blue border) represents the overall attack goal. All leaves of the tree are basic attack steps (colored red or green). 
	The goal is refined into different subtrees using the operator~$\opAnd$, i.e., all of the subtrees need to be successfully attempted for the overall attack to be successful. The overall attack is successful if, the attacker gained the credentials (event~(4)) and successfully entered them (event~(9)). Also for the attack to be successful the defender (node in green) must not (event~(17) with operator~$\opNeg$) have installed additional authentication (event~(6)).
	Event~(4) is refined with the operator~$\opOr$ meaning it is successfully attempted if one of its children (event~(1), (2), (3)) is attempted successfully.
\end{example}


To model complex attack-defense scenarios, a rich set of operators \[\operators = \set{\opAnd,\opOr,\opNeg, \opSand,\opSor,\opTrigger,\opReset}\] is necessary: Intuitively, $\opAnd,\opOr, \opNeg$ behave like their logical counterparts. $\opSand$ and $\opSor$ pose temporal (but not causal) restrictions on the order in which its children are attempted. 
Further, for Trigger~$\opTrigger$,  suppose its child is successfully attempted. In that case, several basic events are activated (\emph{triggered}) for attempting (they cannot be attempted earlier than this). 
Finally, Reset~$\opReset$ is successfully attempted if its input is attempted successfully. As a side effect, events that have been reset can be attempted again, i.e., we allow players to attempt events several times.

A formal semantics of the operators used in the theoretical development in the subsequent section can be found in \Fref{def:powerset} in the Appendix. Formal semantics for the remaining operators are not necessary to follow the paper (They can be found in~\cite{arnold2014,hermanns2016,eisentraut2019}).
%



\section{Tool Overview}
\label{sec:functionality}
This section shows the usage of $\quadtool{}$ both in a graphical user interface and direct command-line access. We evaluate our analyses in \Fref{sec:experiments}.

\begin{figure*}[t]
	\centering
	
	\makebox[\linewidth]{
  \begin{tikzpicture}[auto,node distance=8mm,>=latex,font=\footnotesize]
    
    \node[cost] (s3) {construct or \nodepart{two} import model};

    \node[cost,right=7mm of s3] (s4) {use default values or
      \nodepart{two} add values from various sources};

    \node[cost,right=5mm of s4] (s5) {export to
      model checker \nodepart{two} for further analyses};
    
    \node[cost,above=8mm of s4] (s6) {exact or \PAC{}-analysis
      \nodepart{two} in tool for costs, delays, probabilities};
    
    \node[above left=3mm and 5mm of s3] (s0) {\DOT{}};
    \node[above left=-4mm and 5mm of s3] (s1) {\XML{}};
    \node[below left=-4mm and 5mm of s3] (s2) {\formatstandard{GUI}};
    \node[below left=3mm and 5mm of s3] (s21) {\atbest{}};

    \node[below left=6mm and -10mm of s4] (s11) {iid samples};
    \node[below=6mm of s4] (s12) {time series};
    \node[below right=6mm and -10mm of s4] (s13) {expert estimates};

    \node[above right=0mm and 5mm of s5] (s8) {\prism{}};
    \node[right=5mm of s5] (s9) {\modest{}};
    \node[below right=0mm and 5mm of s5] (s10) {\UPPAAL{}};
	
    \draw[->] (s5) to node[midway,above,sloped]{} (s8);
    \draw[->] (s5) to node[midway,above,sloped]{} (s9);
    \draw[->] (s5) to node[midway,below,sloped]{} (s10);
    
    \draw[->] (s0) to node[midway,above,sloped]{} (s3);
    \draw[->] (s1) to node[midway,above,sloped]{} (s3);
    \draw[->] (s2) to node[midway,below,sloped]{} (s3);
    \draw[->] (s21) to node[midway,below,sloped]{} (s3);
    
    \draw[->] (s3) to node[midway,above,sloped]{} (s4);
    \draw[->] (s4) to node[midway,above,sloped]{} (s5);
    \draw[->] (s4) to node[midway,above,sloped]{} (s6);
    \draw[->,dashed] (s6) to node[above left]{refine} (s3);

    \draw[->] (s11) to node[midway,above,sloped]{} (s4);
    \draw[->] (s12) to node[midway,above,sloped]{} (s4);
    \draw[->] (s13) to node[midway,above,sloped]{} (s4);
  \end{tikzpicture}
}
  \caption{The Workflow with \quadtool{} allows for input from four sources. In the next step the nodes can be refined using quantitative date from different sources. Additionally, the tool allows for analysis of said data. In a last step the data can be exported to model-checkers for further analysis. }
  \label{fig:quadtool}
  \vspace*{-2.5em}
\end{figure*}
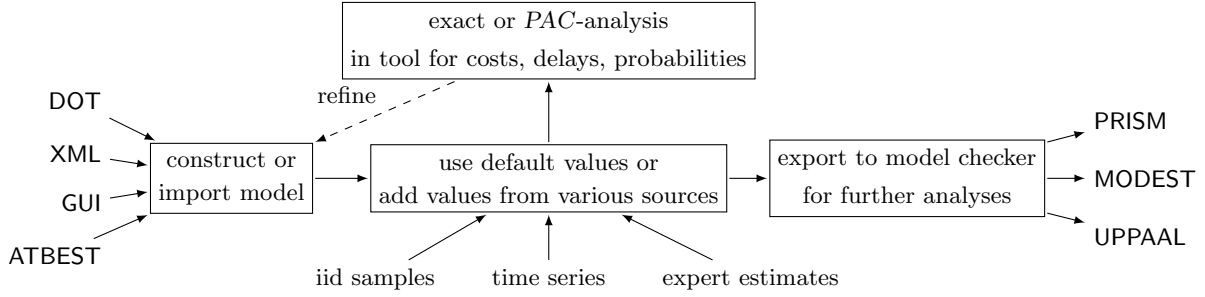 

\subsection{Model Construction}
The workflow using $\quadtool{}$ (see \Fref{fig:quadtool}) starts with
model construction. For convenient modeling, $\quadtool{}$ supports a
dedicated graphical user interface and the following two input
formats:
\begin{itemize}
\item The \DOT{}-format as described in the Graphviz-Tool~\cite{Gansner93atechnique}. We recommend this format as it supports storing all additional parameters like \PAC{}-parameters and simple formatting.
\item The \XML{}-format used in \adtool{}~\cite{gadyatskaya2016}. To ensure compatibility with this tool, our tool only exports strict tree structured $\ADT$ without \PAC{}-values to this format.
\end{itemize}

Our \emph{GUI}, see \Fref{fig:editor}, offers drag-and-drop manipulation of trees, creating new ones from scratch or by import. One can also easily combine existing models by importing several trees and connecting them.

There exist various approaches to attack-defense tree analysis supporting different features (see~\cite{widel2019} for a survey on quantitative attack-defense tree analysis). Hence, $\quadtool{}$ comes with an integrated \emph{feedback system}~(see \Fref{fig:editor}) to support the user in finding potential errors or incompatibilities before exporting. Furthermore, it gives advice on which export formats may be beneficial for the particular model.

\subsection{\atbest{} Benchmark Collection for Attack-Defense Trees}
\label{sec:atbest}
To support the standardization and benchmark collection among attack-defense tree researchers and practitioners, we bring forth an attack-tree {benchmark collection} \emph{AT BEnchmark SuiTe} (\atbest{}) accessible at~\atbestURL{}. By uploading a model and registering with a name and email address, users receive a token to access all uploaded models in the database. In more details, our benchmark set consists of 42 models from 24 previously published papers,~\cite{mauw2005,ten2007,kordy2011,kordy2013,beckers2014,kumar2015,aslanyan2016,pekergin2016,fraile2016,gadyatskaya2016,hermanns2016,kumar2017,ruijters2017,kordy2018,andre2019,fila2019,krichen2019,rosmansyah2020,buldas2020,houmer2020,kim2020b,ohm2020,schurmann2020,roblesmodel}, which we constructed using \quadtool{}. Additionally, we generated 626 models of various sizes for performance testing. These are binary \adt{} using the gates $\opAnd$ and $\opOr$, created by combining two smaller models with a new root recurrently. 

\subsection{Obtaining \PAC{}-Quantities for Attack-Defense Trees}
\label{sec:getpacqadt}
After creating a model, we add default values or estimates to basic events (second step in the workflow (see \Fref{fig:quadtool})). Our \PAC{} (s. \Fref{sec:pac}) assumption is less restrictive than the standard assumption of exact quantities and already applicable to various examples.


To derive \PAC{}-parameters, the simplest way is to calculate a Gaussian estimate with confidence intervals. The Gaussian estimate is best for independent identically distributed (i.i.d.)  samples. Our tool allows the generation of those values via csv-input. One can also directly use a AutoRegressive-Moving Average (\formatnames{ARIMA}) model~\cite{arimacovid,arimaelectricity,arimafinance,arimapacarticle,arimarental,hughes2010} inside the tool to derive bounds for cases with time-series data.  

%
%


Please note that we do not aim to create the best possible time-series model for a given data-set using our tool since this still requires a more detailed analysis. However, we want to allow users to get a reasonably good prediction quickly.

\subsection{Analyses}
\label{sec:pacanalysis}
\quadtool{} allows to analyze models using an exact or \PAC{} analysis or using existing model checkers. 

 Like the \adtool{}, $\quadtool$ supports direct analyses on the tree for exact minimal attack costs, exact minimal attack delays, and exact success probabilities. Additionally, we support our new \PAC{}-input analyses for each of these domains. 

Furthermore, users can export the final model to various other tools and model checkers using one of the following file formats: \XML{}, \DOT{}, \prism{}-, \modest{}-, \UPPAAL{}-input-language. For \prism{}, we use the semantics in terms of stochastic turn-based games given in~\cite{eisentraut2019} and for
\modest{} and \UPPAAL{} the semantics in terms of stochastic timed automata given in~\cite{hermanns2016}. \footnote{To use \storm{}, we suggest using the export to \modest{} first and then, use their export to the \jani{}-format.}  In the following, we describe the interfaces to other tools and their use in more detail.

\para{Connection to \adtool{} and \attop} These two attack-defense tree tools have a restricted set of features and expect \XML{} inputs. Thus, we cannot analyze every attack-defense tree constructed in \quadtool{} using these tools. For instance, whenever a model has multiple root nodes, our export omits all but one tree. Furthermore, the \adtool{} only supports the operators $\opOr$, $\opAnd$, $\opSand$ and no timed events (with clocks or distributions). 

\para{Timed Analysis} Analysis involving time can be either done within $\quadtool{}$ or outsourced to \modest{} or \UPPAAL{}.

The translation to \modest{} uses actions for each three-valued logic state, i.e., each basic event can either be undecided, successfully or unsuccessfully attempted. For \UPPAAL{} a similar approach with automata for every state is used. Both approaches can model systems with multiple root nodes and complex operators.
%
%

\para{Game Analysis} Lastly, $\quadtool{}$ can also export to the \prism{} format, not supporting time, but supporting games.  In this framework, all basic events have the same chance to be executed at a given (discrete) point in time, impacting the evaluation of sequential operators.

\para{Integration} We provide the functionality to send a model directly to \UPPAAL{}~cora, \adtool{}, \prismgames{} and \modest{}~simulator in the GUI, after specifying the Tool location, as seen in \Fref{fig:editor} (right lower panel).

\begin{figure*}[t]
	\centerline{\includegraphics[height=6.5cm]{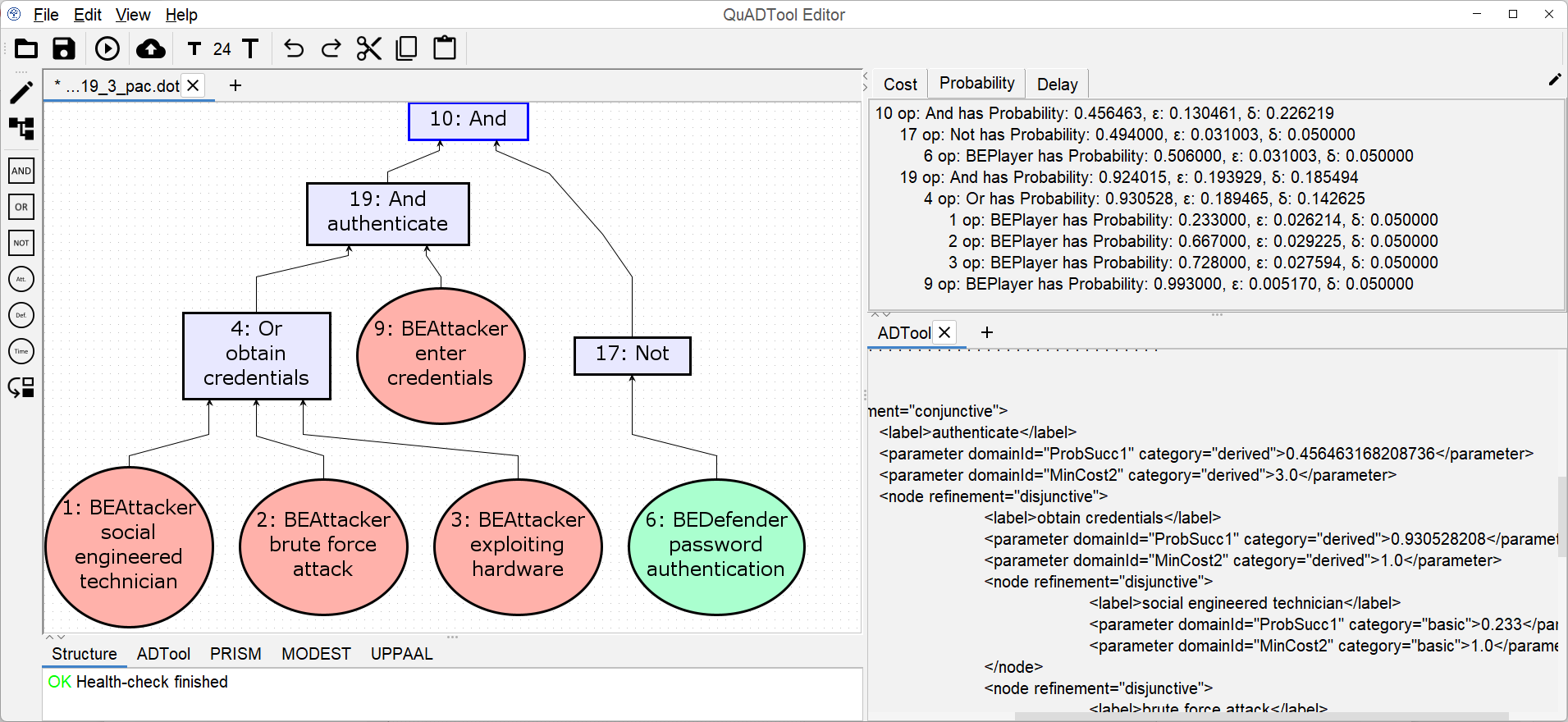}}
	\caption{$\quadtool$ with the case study from~\cite{fila2019}. At the bottom, one can see the feedback for the selected export or general issues. On the right, the quantitative results for the selected domain (here probabilities) are given for all nodes with the \PAC{}-quantities. Below is the result of a model checker (here \adtool{}).}
	\label{fig:editor}
\end{figure*}

\section{\PAC{}-Quantitative \adt{}: Fighting the Uncertainty}
\label{sec:pac}
In this section we describe the technique underlying our tools' functionality to analyze $\adt$ with not fully known quantities. While it is less standard to do so in a tool paper, this is not described elsewhere and thus must be included in the tools' contributions for completeness here. 

A standard way to reason about costs, probabilities, and time within attack-defense trees is to equip basic events with quantitative information and use bottom-up traversal algorithms to derive the root's final value~\cite{mauw2005,kordy2011,arnold2014,jhawar2015}. In this section, we augment this analysis to bound the uncertainty resulting from estimating the quantitative values of the basic events.

In \Fref{app:correctness}, we show that our computations of
probabilities, costs and delays are correct with reference to standard Boolean
and powerset valuations of attack-defense
trees~\cite{mauw2005,kordy2011,gadyatskaya2016}. However, standard
bottom-up traversal algorithms are efficient and correct only on
simpler structures than our \adt{} in \Fref{def:add} (for more
details, see~\cite{kordy2018}). In particular, they require a
\emph{static interpretation}, i.e., there is only one shot to decide
which basic events players try and which they do not
try. Consequently, we restrict ourselves to one of two domains supported by the
\adtool{} (attack-defense trees with gates $\opAnd$, $\opOr$ and
$\opNeg$). For reader's convenience, we exclude $\opTrigger$ in this
section but include it in \Fref{app:correctness}. Hence, our analysis also covers
the second domain of the \adtool{}.

\subsection{Quantitative \adt{} and \PAC{}-Quantitative \adt{}}
\para{Quantitative \adt{}.} Given a type $\type$ of the quantitative information, the quantitative \adt{} is an \adt{} equipped with a \emph{quantitative valuation of basic events} $\valuesBE: \basicEvents \to \type$.  The definition is generic to carry over to various settings: While real numbers may represent cost~\cite{gadyatskaya2016}, acylic phase-type distributions may, for instance, represent the likelihood of a successful attack over time~\cite{arnold2014}.

\begin{example}[Type]
  For success probabilities, we simply take $\type=[0,1]$. For costs and delays, we take $\type=\RR_{\geq 0}$.
\end{example}

\para{\PAC{}-Quantities.} In most real-world examples, the quantitative information attached to basic events is rarely known precisely but comes with uncertainty. We call quantitative information \emph{probably approximately correct} (\PAC{}), if there is an \emph{imprecision bound} $\error$ and an \emph{uncertainty probability} $\prob$, such that the distance between the true value and the given estimate is less than $\error$ with probability at least $1-\prob$. Formally:

\begin{definition}[\PAC]
  Let $\type$ be a set equipped with a metric~$\metric$. We say that a $\type$-valued random variable $X$ is $(\error,\prob)$-\PAC{} for the true value~$\val$ if the probability\footnote{The method used to learn the estimate $X$ from data sampled from a distribution parametrized by $\val$ usually determines the probability space. For instance, for a $\val$-biased coin with $\val=0.49$ and 20 samples (coin tosses),
    the standard statistics methods yield a probability space with highest density of $X$ in $0.49$ and inducing the respective confidence intervals for $X$ depending on the outcome of the experiment.} of $X \in \set{t \mid \metric(t,\val)\le \error}$ is (at least) $1-\prob$. Exact values can be considered $(0,0)$-\PAC{}.
\end{definition}

\begin{example}[Metric]
  For probabilities, costs, and delays, the metric is the standard Euclidean distance (absolute value of their difference).
\end{example}	

\para{\PAC{}-Quantitative \adt{}} Given a type $\type$ of the quantitative information, the \PAC{}-quantitative \adt{} is a quantitative \adt{} equipped additionally with \emph{\PAC{}-parameters} $\error[\basicEvent] \in \RR_{\geq 0}$ and $\prob[\basicEvent] \in [0,1]$ for each basic event $\basicEvent\in\basicEvents$.

\subsection{\PAC{}-Input Extension of Bottom-Up Analyses}
Before giving the generic algorithm for the bottom-up propagation of quantitative values~\cite{mauw2005,kordy2011}, we will give a short introduction to bottom-up analyses.

\para{Bottom-Up Traversal Analysis} Let $\ADT$ be an attack-defense tree with set $\vertices$ of nodes, $\topGoal$ its root, $\type$ a non-empty domain of quantitative information, $\interpretation$ an interpretation i.e. a mapping of \adt{}-operators to operations over $\type$, and $\valuesBE$ a basic-events valuation with input values from $\type$. Then its bottom-up traversal semantics $\values \colon \vertices \to \type$ for $\ADT$ is defined recursively. The valuation of $\ADT{}$ is defined as the value of its root node $\values(\topGoal)$.

\begin{example}[Probabilistic \adt{}] \label{def:prob} We have already seen that the domain for probabilities is $[0,1]$.  Here we instantiate the interpretation for probabilities. We define a probabilistic interpretation~$\interpretation$ (with $\type=[0,1]$) as follows: (Similar examples for the domain of Cost and Probability can be found in \Fref{app:examplesBottomUp})
	\begin{itemize}
		\item
		$\displaystyle{\interpretation(\vertex) (x_1,x_2)= x_1 \cdot x_2}$, if 
		$\typeFunction(\vertex)=\opAnd$
		\item
		$\displaystyle{\interpretation(\vertex)(x_1,x_2)= x_1 + x_2 - x_1 \cdot x_2}$ for $\typeFunction(\vertex)=\opOr$
		\item $\displaystyle{\interpretation(\opNeg)(x)=1-x}$
	\end{itemize}	
\end{example}

\para{Extension to PAC-Values} In the following, we devise an algorithmic approach that allows us to derive a \PAC{}-aware variant from almost any bottom-up traversal analysis (according to its previous definition)\footnote{In fact, the  approach is independent of the concrete interpretation $\interpretation$ chosen, as long as for $\RR_{\geq0}$-interpretation it involves solely the operations $+$, $-$, $\max$, $\min$ and $\cdot$, or \emph{any combination} thereof.} as long as it holds $\type \subseteq \RR_{\geq 0}$ for the domain~$\type$. To this end, we use the notation $x_i \pm \error[i] \coloneqq [\max(0,x_i-\error[i]),x_i+\error[i]]$\footnote{For success probabilities, we set $x_i \pm \error[i] \coloneqq
  [\max(0,x_i-\error[i]),\min(x_i+\error[i],1)]$}.

We propagate the uncertainty along with the nodes during the bottom-up traversal analysis resembling interval arithmetics from canonical error analysis. Whenever the interpretation uses one of the operators $+$, $-$, $\cdot$, $\max$, $\min$ over $\RR_{\geq 0}$, for its \PAC{}-aware variant, we additionally compute a new uncertainty bound and uncertainty probability for the resulting value. The nodes' uncertainty bound and uncertainty probability are solely based on its operands' uncertainty bounds and uncertainty probabilities.

\para{\PAC{}-Bounds Propagation} Without any assumptions on the distribution of the uncertainty, operations are as follows:
  Given values $x_1, x_2 \in [0,1]$ that are $(\error[i],\prob[i])$-\PAC{}, the uncertainty of
  \begin{enumerate}[label=(\subscript{R}{{\arabic*}}),itemindent=1em]
  \item $1-x_i$ can be bound by $\error = \error[i]$.
  \item $x_1 \cdot x_2$ can be bound by $\error = x_1\cdot \error[2] + x_2 \cdot \error[1] + \error[1]\cdot\error[2]$.
  \item $x_1 + x_2 - x_1 \cdot x_2$ can be bound by  $\error = \error[1] + \error[2] + x_1\cdot \error[2] + x_2 \cdot
    \error[1] + \error[1]\cdot\error[2]$
  \item $\max(x_1,x_2)$ can be bound by $\error=\max(\error[1],\error[2])$
  \item $\min(x_1,x_2)$ can be bound by $\error=\max(\error[1],\error[2])$
  \end{enumerate}
  For all operations the uncertainty probability $\delta$ is (at most) $1 - (1-\prob[1]) \cdot (1-\prob[2]) $, if all $x_i$ are independent, and otherwise (at most) $\prob[1] + \prob[2]$. As probability analysis relies on independence of the basic events, we will use the first bound in the following.

  \medskip\noindent In summary, the result of each of the given operations is $(\error,\prob)$-\PAC{}. Furthermore, each constant~$c \in \RR$ is (0,0)-\PAC{}.

\begin{example}[PAC-Bounds Computation for Probabilities]
  Let $\ADT$ be an attack-defense tree with set $\vertices$ of nodes, let $\values[\basicEvents]$ be a \PAC{} probability valuation for basic events, i.e., for every basic event~$\basicEvent \in \basicEvents$, it holds $\values[\basicEvents](\basicEvent)$ is $(\error[\basicEvent],\prob[\basicEvent])$-\PAC{} for some $\error[\basicEvent] \in \RR_{\geq 0}$ and $\prob[\basicEvent] \in [0,1]$. The  uncertainty for $\vertex \in \vertices$ is computed as follows:
  \begin{enumerate}
  \item $\vertex \in \basicEvents$ and $\vertex$ is not triggerable: return $(\error[\basicEvent],\prob[\basicEvent])$.
  \item $\typeFunction(\vertex)=\opAnd$ and $\inp[1],\inp[2]$ are its inputs: recursively compute $(\error[\inp[1]],\prob[\inp[1]])$ and $(\error[\inp[2]],\prob[\inp[2]])$. Use rule for multiplication \subscript{R}{2}, i.e., return  $\error= \values(\inp[1]) \cdot \error[\inp[2]] + \values(\inp[2]) \cdot \error[\inp[1]] + \error[\inp[1]]\cdot\error[\inp[2]]$ and $\prob = 1 - (1-\prob[\inp[1]]) \cdot (1-\prob[\inp[2]]) $
  \item $\typeFunction(\vertex)=\opOr$ and $\inp[1],\inp[2]$ are its inputs: recursively compute $(\error[\inp[1]],\prob[\inp[1]])$ and $(\error[\inp[2]],\prob[\inp[2]])$. Use rule \subscript{R}{3}, i.e., return
    $\error= \error[\inp[1]]+ \error[\inp[2]] + \values(\inp[1]) \cdot \error[\inp[2]] + \values(\inp[2]) \cdot \error[\inp[1]] + \error[\inp[1]]\cdot\error[\inp[2]]$ and $\prob = 1 - (1-\prob[\inp[1]]) \cdot (1-\prob[\inp[2]]) $
  \item $\typeFunction(\vertex)=\opNeg$ and $\inp$ is its input: recursively compute $(\error[\inp],\prob[\inp])$. Use rule \subscript{R}{1}, i.e., return $\error= \error[\inp]$ and $\prob = \prob[\inp]$
  \end{enumerate}
\end{example}
Formally, we can state that the resulting interval is the tightest reflecting the input intervals: Let $\interpretation(\vertex) \in x_{\vertex} \pm \epsilon_{\vertex}$ for each $\vertex \in \basicEvents$ and $\gate \in \gates$ then $\interpretation(\gate)$ lies in the interval computed above. Vice versa, for every $\interpretation(\gate)$ in the interval, there are $\interpretation(\vertex) \in x_{\vertex} \pm \epsilon_{\vertex}$ (consistent with $\interpretation(\gate)$). Furthermore, with $\delta_{\vertex}$ denoting the probability that $\interpretation(\vertex) \notin x_{\vertex} \pm \epsilon_{\vertex}$, the probability, that the resulting value is outside of the computed \PAC~interval, is at most the computed $\delta$.

In \quadtool{}, we also provide \PAC{}-aware bottom-up analyses for delays and costs. The proof of correctness can be found in \Fref{app:pac}.

\begin{example}[Propagation of Uncertainty]
	In practice, $\delta = 0.05$ or $\delta=0.01$ seem reasonably small uncertainty probabilities for meaningful \PAC{} data (since they are most commonly used significance criteria). As most examples from literature have less than 30 nodes (s. \Fref{sec:experiments}), a corresponding binary tree would have a depth of 5.  This would lead to $\delta \approx 0.56$ or  $\delta \approx 0.15$ for an inital $\delta$ of $0.05$ or $0.01$ at the root node respectively. Thus, selecting data of high confidence i.e. low $\delta$ is crucial for good overall estimates.
\end{example}
%
%

\begin{figure*}[t]
	\centering
	\begin{subfigure}{0.35\linewidth}	
		\begin{itemize}[label={}]
			\item ID 10 $p: 0.4594491$
			\item \begin{itemize}[label={}]
				\item ID 19 $p: 0.9188982$
				\item \begin{itemize}[label={}]
					\item ID 4 $p: 0.92818$
					\item \begin{itemize}[label={}]
						\item ID 1 $p: 0.24$
						\item ID 2 $p: 0.65$
						\item ID 3 $p: 0.73$
					\end{itemize}
					\item ID 9 $p: 0.99$
				\end{itemize}
				\item ID 17 $p: 0.5$
				\item \begin{itemize}[label={}]
					\item ID 6 $p: 0.5$
				\end{itemize}
			\end{itemize}
		\end{itemize}
		\vspace{-1.5em}
		\caption{}
		\label{fig:subfigB}
	\end{subfigure}
	\begin{subfigure}{0.64\linewidth}
		\begin{itemize}[label={}]
			\item ID 10 $p: 0.456463\;\epsilon: 0.130461\;\delta: 0.226219$
			\item \begin{itemize}[label={}]
				\item ID 19 $p: 0.924015\;\epsilon: 0.193929\;\delta: 0.185494$
				\item \begin{itemize}[label={}]
					\item ID 4 $p: 0.930528\;\epsilon: 0.1894165\;\delta: 0.142625$
					\item \begin{itemize}[label={}]
						\item ID 1 $p: 0.233\;\epsilon: 0.026214\;\delta: 0.05$
						\item ID 2 $p: 0.667\;\epsilon: 0.029225\;\delta: 0.05$
						\item ID 3 $p: 0.728\;\epsilon: 0.027594\;\delta: 0.05$
					\end{itemize}
					\item ID 9 $p: 0.993\;\epsilon: 0.005170\;\delta: 0.05$
				\end{itemize}
				\item ID 17 $p: 0.494\;\epsilon: 0.031003\;\delta: 0.05$
				\item \begin{itemize}[label={}]
					\item ID 6 $p: 0.506\;\epsilon: 0.031003\;\delta: 0.05$
				\end{itemize}
			\end{itemize}
		\end{itemize}
		\vspace{-1.5em}
		\caption{}
		\label{fig:subfigC}
	\end{subfigure}
	\caption{Example evaluation of the model in \Fref{fig:toyexample} (a) without \PAC{} values and (b) with \PAC{}-values. The \PAC{}-values were generated by sampling from Bernoulli distributions parameterized by the exact value.}
	\label{fig:examplevaluation}
\end{figure*}

\subsection{Analysis Example}
We use one of the models introduced by Fila and Widel in \cite{fila2019}(s. \Fref{fig:toyexample}) about manipulating a power meter as example for the analysis workflow with \PAC{} values. 
We will use the domain of probability and assign each of the basic events a success probability from the paper (see \Fref{fig:examplevaluation} (a)). For the defender node with id 6 we set the probability to 50\% to represent, that only half of the power meters have this function enabled.

Using the exact properties the probability for a successful attack is about 46\% (see \Fref{fig:examplevaluation} (a)). To generate the \PAC{} values we use R 4.3.1 to sample a Bernoulli distribution, with the exact values as parameter, 1000 times. From these values we use the built-in calculation of our tool to retrieve the \PAC{} values using a Gaussian-Estimator.

The propagated values are displayed in \Fref{fig:examplevaluation} (b). The calculations in the tool are done with arbitrary precision. The intermediate values are only rounded for display. Due to the initial confidence level of 5\% the confidence level of the overall analysis is about 77\%. Taking the $\epsilon$ in account the analysis returns an estimate of the true probability with confidence interval of about 45.65\% [33.61\%-57.69\%]. So the true value calculated in the exact example lies inside the bound generated by our \PAC{}-propagation. 

This extended analysis does, however, come with an increased computational effort, which still is inside acceptable bounds (see \Fref{sec:experiments}).


\section{Experiments}
\label{sec:experiments}
We discuss the performance of~$\quadtool$ on our collection of benchmarks found on \atbest{}.

For performance measurements, we used 42 security models, ranging from 5 to 30 nodes, from existing literature (s. \Fref{sec:atbest}), allowing the whole operator set. 17 of them contained operators, that were incompatible with \XML{}. Therefore, they were only used for the remaining formats. We executed the conversion 50 times each on a computer with Intel\textregistered Core\texttrademark i7 processor and 16 GB RAM.  


\begin{figure*}[t]
	\vspace*{-1em}
	\centerline{\includegraphics[width=\textwidth]{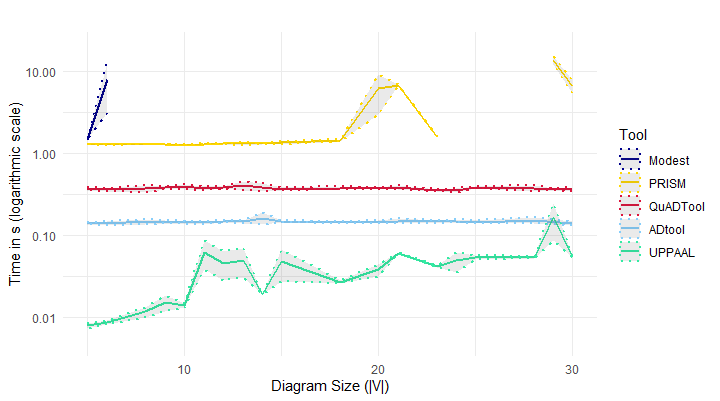}}
	\caption{95\%-confidence intervals of time (logarithmic scale) necessary for the corresponding model checker to evaluate models (for the success probability) depicted by the ADT size.
	Most executions take less than 10 seconds. Y-position of line start corresponds with position in legend. As most models from \atbest{} contain mostly player nodes, our timed analysis using \modest{} returns stack overflow errors on models of $size \geq 5$. \prismgames{} runtime seems to be more dependent on model composition, rather than just the size. The evaluation time for \quadtool{}, \adtool{} and \UPPAAL{} Cora are all below 1 second.}
	\label{fig:tooleval}
	\vspace*{-1em}
\end{figure*}

\begin{figure*}[t]
	\centerline{\includegraphics[width=\textwidth]{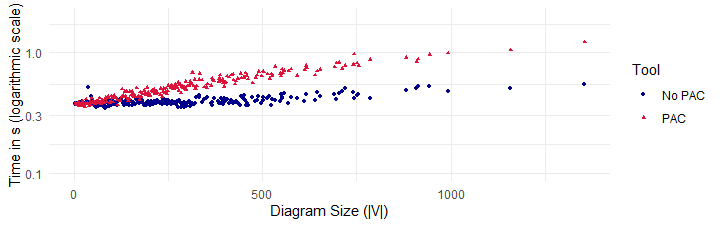}}
	\caption{Model checking times (mean per size) \PAC{} (red triangles) vs. non-\PAC{} (blue dots) models in \quadtool. The highest execution time was 1.2s for 1355 nodes. }
	\label{fig:toolspac}
	\vspace*{-1em}
\end{figure*}

In the next step, we conducted experiments on how long the model checkers \adtool{}, \UPPAAL{}, \prismgames{} and \modest{} need to analyze the \adt{} success and failure probability on the models previously used for conversion. Every execution was repeated five times.

Lastly, we compare the execution time of \PAC{} analysis with only propagating the probability. Thus, we used 636 model generated randomly only having inner nodes with $\opAnd$ and $\opOr$ with a size of up to 1355 nodes, since literature examples are limited in amount and size. We calculate the success probability once with and once without \PAC{} values.



 \para{Results} are depicted in \Fref{fig:tooleval} and \ref{fig:toolspac}.  Our experiments show the following:
 \begin{itemize}
 \item Exports to the model checkers
   \UPPAAL{}, \modest{}, \prism{} and to the \XML{} format are
   feasible for models of sizes commonly found in literature. With all except \prism{} needing less then a second for each model (see \Fref{app:conversionTimes}).
 \item Generating the stochastic game explicitly, which the
   translation to \prism{} relies on, is time-consuming for models
   with more than~30 nodes (see \Fref{app:conversionTimes}). Yet, verification using \prism{} may
   still be of use in practice. 
 \item The execution time may depend more on the model composition
   than on the actual size of the underlying tree (see
   \Fref{fig:tooleval}, \ref{fig:toolspac}).
 \item Verification using the \adtool{} and \UPPAAL{} is faster than using \prismgames{} or \modest{}. If evaluation is possible, the results are returned in less then ten seconds (see \Fref{fig:tooleval}). 
 \item  \UPPAAL{} model checking is useful for more complex objectives, posing an appropriate trade off. 
 \item For \PAC{} vs.\ non-\PAC{} analysis in \quadtool, the slope of runtime differs roughly by a factor 10. None of the approaches takes more than a few seconds  (see \Fref{fig:toolspac}) even for attack-defense trees with more than 1000~vertices.
 \item The runtime of \adtool{} and \quadtool{} using bottom-up analysis are in the same order in our experiments (see \Fref{fig:tooleval}).
\end{itemize}



\section{Conclusion and Future Work}\label{sec:conclusion} 
We have provided a simple tool for synthesis and analysis of attack-defense trees that allows connection to model-checker workflows. The CLI enables further integration into other software pipelines. Settings for various degrees of background knowledge allow also non-experts to access the field of $\ADT{}$.
Furthermore, we provide a simple but effective way to bound the uncertainty in quantitative attack-defense tree analysis, using probably approximately correct inputs \emph{and} outputs, obtaining uncertainty from data. 
Our experiments show that all analysis can be done in reasonable time even for models much larger than found in literature. 
Altogether, \quadtool{} presents several steps making the $\ADT{}$ framework more applicable in practice.

\para{Future Work.} 
Firstly, we aim to support larger models and models with heterogeneous subtrees.  To this end, we want to design and implement compositional verification techniques to efficiently analyze, for instance, models where one subtree consists of only $\opAnd$ and $\opOr$ and one player, another requiring full-fledged timed analysis and a third one games.

Secondly, the semantics in terms of stochastic timed automata and stochastic turn-based games suffer from a  blow-up since all sequences of events need to be taken into account. One could employ partial-order reduction techniques in the conversion before model checking to reduce the model size.
%

\para{Acknowledgments.} We want to thank everyone who has contributed case studies to our benchmark collection \atbest{} so far.
This research was funded in part by the German Research Foundation (DFG) project No.~427755713 \emph{GOPro}, the MUNI Award in Science and Humanities {MUNI/I/1757/2021} of the Grant Agency of Masaryk University, the  German Academic Scholarship Foundation and the ProSec project.


\appendix
\pagebreak

\section{Quantitative Bottom-Up Analyses for Cost and Delays}
\label{app:examplesBottomUp}
\begin{example}[Delay \adt{}]
	\label{def:delay}
	We define both the minimal and maximal delay interpretation~$\interpretation$ with $\type=\RR_{\geq0}\times\RR_{\geq0}$, where the first component represents the delay of succeeding, the second component the delay of failing
	\begin{itemize}
		\item
		$\displaystyle{\interpretation(\vertex)((x_1,x_2),(y_1,y_2)) = (\max(x_1,y_1),\min(x_2,y_2))}$ (for minimal delay) and\\ $(\max(x_1,y_1).\max(x_2,y_2))$ (for maximal delay), if 
		$\typeFunction(\vertex)=\opAnd$.
		\item
		$\displaystyle{\interpretation(\vertex)((x_1,x_2),(y_1,y_2))= (\min(x_1,y_1),\max(x_2,y_2))}$ (for minimal delay) and \\ $\displaystyle{\interpretation(\vertex)((x_1,x_2),(y_1,y_2))= (\max(x_1,y_1),\max(x_2,y_2))}$ (for maximal delay) \\ if $\typeFunction(\vertex)=\opOr$
		\item $\displaystyle{\interpretation(\opNeg)(x_1,x_2)=(x_2,x_1)}$
	\end{itemize}
\end{example}

\begin{example}[Cost \adt{}]
	\label{def:cost}
	We define both the minimal and maximal cost interpretation~$\interpretation$ with $\type=\RR_{\geq 0}\times \RR_{\geq0}$, where the first component represents the cost of succeeding, the second component the cost of failing.
	\begin{itemize}
		\item
		$\displaystyle{\interpretation(\vertex)((x_1,x_2),(y_1,y_2))= (x_1+y_1, \min(x_2,y_2))}$ (for minimal cost) and $(x_1+y_1,\max(x_2,y_2))$ (for maximal cost), if 
		if $\typeFunction(\vertex)=\opAnd$
		\item
		$\displaystyle{\interpretation(\vertex)((x_1,x_2),(y_1,y_2))=(\min(x_1,y_1),x_2+y_2)}$ (for minimal cost) and   \\ $\displaystyle{\interpretation(\vertex)((x_1,x_2),(y_1,y_2))=(\max(x_1,y_1),x_2+y_2)}$ (for maximal cost) if $\typeFunction(\vertex)=\opOr$
		\item $\displaystyle{\interpretation(\opNeg)(x_1,x_2)=(x_2,x_1)}$
	\end{itemize}
\end{example}

\section{Conversion Times}
\label{app:conversionTimes}
\begin{figure*}[htbp]
	\centerline{\includegraphics[width=\linewidth]{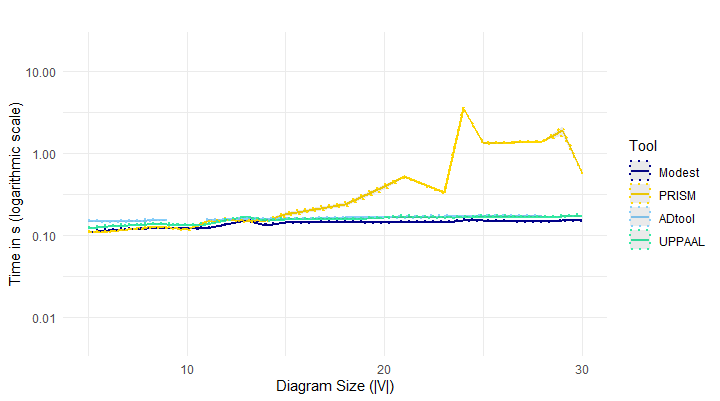}}
	\caption{Execution times for conversion of \DOT{}-files to the other formats.}
	\label{fig:conversion}
\end{figure*} 

\section{Correctness of Bottom-Up Traversal Semantics}
\label{app:correctness}
The following section recalls definitions from~\cite{kordy2011} and
adapts them to our attack-defense tree setting. \emph{We only allow
  trees here (no DAGs).}  We include the operator~$\opTrigger$ only if
every basic event is triggered by at most one vertex. Additionally, we
require that vertices labeled with $\opTrigger$ have no successors,
and the main tree and the subtrees rooted in vertices with
$\opTrigger$ to be mutually disjoint. In this case, we can ensure an
order satisfying the constraints imposed by~$\opTrigger$ for each
satisfying assignment.  In contrast to~\cite{kordy2011}, we allow
$\opTrigger$ as a restricted variant of $\opSand$\footnote{We use
  $\opSand$ to represent a timed order and $\opTrigger$ to represent a
  logic order.} and use $\opNeg$ rather than player swapping to focus
on the player \emph{controlling} an event rather than whether an event
is \emph{useful} or \emph{hindering} for an attack. Additionally, we
assume that every basic event can only occur once in the tree. Thus,
we restrict the multiset semantics of~\cite{mauw2005,kordy2011} to a
powerset semantics.

\begin{definition}[Extension of Example~5 -~7]
  In each case, $\interpretation(\opTrigger) = \Id$ and
  $\interpretation(\basicEvent) = \interpretation(\opAnd)$ for
  triggerable basic events $\basicEvent \in \basicEvents$.
\end{definition}

\noindent Let $X$ be a set and let $X_1, X_2 \subseteq \powerset{X}$,
i.e. $X_1$ and $X_2$ are sets of sets with elements in $X$. We define
$X_1 \otimes X_2 \coloneqq \bigcup_{(x_1,x_2) \in X_{11} \times
  X_{21}} \set{x_1\cup x_2}$. To restrict a set of basic events~$X$ to
the basic events occurring in the subtree rooted in a
vertex~$\vertex$, we write~$X|_{\vertex}$.

\begin{definition}[Powerset Attack-Defense Trees]
  \label{def:powerset}
  We consider the operators $\opAnd$,$\opOr$, $\opTrigger$ and
  $\opNeg$ and let all basic events be in
  $\basicAttacks \disjointUnion \basicDefenses$. We define a powerset
  interpretation~$\interpretation$ with
  $\type=\powerset{\powerset{\basicEvents}} \times \powerset{\powerset{\basicEvents}}$.
  \begin{itemize}
  \item
    $\interpretation(\vertex)((X_{11},X_{12}),(X_{21},X_{22}))=
      (X_{11}\otimes X_{21},$ $X_{21} \cup X_{22} \cup (X_{21} \otimes
      X_{22}) \cup (X_{12} \otimes X_{21}) \cup (X_{11} \otimes X_{22}))$ if $\vertex$ is a
    triggerable basic event \emph{or} $\typeFunction(\vertex)=\opAnd$
  \item
    $\interpretation(\vertex)((X_{11},X_{12}),(X_{21},X_{22}))=
      ((X_{11} \cup X_{21} \cup (X_{11} \otimes X_{21}) \cup$ $(X_{12}
      \otimes X_{21}) \cup (X_{11} \otimes X_{22})), (X_{12}\otimes
      X_{22}))$ for $\typeFunction(\vertex)=\opOr$
  \item $\displaystyle{\interpretation(\opNeg)((X_1,X_2))=(X_2,X_1)}$ 
  \item $\displaystyle{\interpretation(\opTrigger)=\Id}$
  \end{itemize}
  We denote the powerset bottom-up traversal semantics by~$\values<P>$.
\end{definition}

The powerset bottom-up traversal semantics computes two sets of
sets. In the first component, we collect all satisfying
assignments, in the second one all unsatisfying
assignments. Please note that we only denote events, which need to be
set to true. All events not contained, need to be set to false.

To compute successful and unsuccessful attacks at once, we
need to set
$\values[\basicEvents]<P>(\basicEvent) \coloneqq
(\set{\set{\basicEvent}},\set{\emptyset})$. For the powerset bottom-up
traversal semantics, we implicitely assume this valuation for basic
events in the following.

\begin{definition}[Boolean Attack-Defense Trees]
  \label{def:boolean}
  We consider the operators $\opAnd$,$\opOr$, $\opTrigger$ and
  $\opNeg$ and let all basic events be in
  $\basicAttacks \disjointUnion \basicDefenses$. We define a Boolean
  interpretation~$\interpretation$ with
  $\type=\set{\true,\false}$
  \begin{itemize}
  \item
    $\displaystyle{\interpretation(\vertex) (x_1,x_2)= \land
    }$, if $\vertex$ is a triggerable basic event \emph{or}
    $\typeFunction(\vertex)=\opAnd$
  \item $\displaystyle{\interpretation(\vertex)=\lor}$ for
    $\typeFunction(\vertex)=\opOr$
  \item $\displaystyle{\interpretation(\opNeg)=\neg}$
  \item $\displaystyle{\interpretation(\opTrigger)=\Id}$
  \end{itemize}
  We denote the Boolean bottom-up traversal semantics by~$\values<B>$.
\end{definition}

Let $X \subseteq \basicEvents$. In the following, we denote by
$\values[\basicEvents,X]<B>$ the Boolean valuation such that
$\values[\basicEvents,X]<B>(\basicEvent)=\true$ if $\basicEvent \in X$
and false, otherwise.

\begin{theorem}[Powerset and Boolean Attack-Defense Trees]
  \label{theo:correctness}
  Let $\exampleADD$ be an attack-defense tree and let
  $\values<P>(\vertex) = (X_1,X_2)$.
  \begin{enumerate}
  \item $x_1 \in X_1$ iff $\values<B>(\vertex)=\true$
    w.r.t. $\values[\basicEvents,x_1]<B>$
  \item $x_2 \in X_2$ iff $\values<B>(\vertex)=\false$
    w.r.t. $\values[\basicEvents,x_2]<B>$
  \end{enumerate}
\end{theorem}

\begin{proof}
  Let $\exampleADD$ be an attack-defense tree. We proof both claims
  simultaneously by induction on the structure of the attack-defense
  tree.
  \begin{itemize}
  \item $\vertex$ is a non-triggerable basic
    event. By definition $\values<B>(\vertex) = \true$ iff
    $\values[\basicEvents]<B>(\vertex)=\true$. Since
    $\values[\basicEvents]<P>(\vertex) = (\set{\set{\vertex}},
    \set{\emptyset})$, the claims holds true.
  \item \emph{Induction Hypothesis:} For all subtrees of $\vertex$
    rooted in $\vertex<\ast>$ holds:
    \begin{enumerate} 
    \item $x_1 \in X_1$ iff $\values<B>(\vertex<\ast>)=\true$
    w.r.t. $\values[\basicEvents,x_1]<B>$
  \item $x_2 \in X_2$ iff $\values<B>(\vertex<\ast>)=\false$
    w.r.t. $\values[\basicEvents,x_2]<B>$
    \end{enumerate}

  \item $\typeFunction(\vertex)=\opTrigger$: The claim holds by
    induction hypothesis and the definitions of $\values<P>$ and
    $\values<B>$ (identity function for vertices labelled by
    $\opTrigger$).
  \item $\vertex \in \basicEvents$ and triggered by
    $\vertex<\ast> \in \vertices$. Let
    $\values<P>(\vertex<\ast>)=(X_1^\ast,X_2^\ast)$. We have
    $\values[\basicEvents]<P>(\vertex)=(\set{\set{\vertex}},\set{\emptyset})$
    and
    $\values<B>(\vertex) = \values<B>(\vertex<\ast>) \land
    \values<B>(\vertex)$. Thus, every satisfying assignment of the
    subtree rooted in $\vertex<\ast>$ extended by $\vertex$ set to
    true is a satisfying assignment of $\vertex$. We apply the
    induction hypothesis to the subtree rooted in
    $\vertex<\ast>$. Hence, by definition of $\otimes$ and the
    induction hypothesis, Claim~(1) holds true. Additionally,
    $\values<B>(\vertex) = \false$ if $\vertex$ is set to false
    (i.e. $\values[\basicEvents]<B>(\vertex)=\false$) or
    $\values<B>(\vertex<\ast>)=\false$. Hence, by induction hypothesis
    and the definition of $\values<P>$ for triggerable basic events,
    Claim~(2) holds, too.
  \item $\typeFunction(\vertex)=\opAnd$: Analogously to the proof for
    triggerable basic events. However, the induction hypothesis needs
    to be used on both subtrees.
  \item $\typeFunction(\vertex)=\opOr$: Analogously to the proof for
    triggerable basic events. However, the induction hypothesis needs
    to be used for both subtrees and the proofs of the two claims need
    to be swapped.
  \item $\typeFunction(\vertex)=\opNeg$: Let $\vertex<\ast>$ be the input of
    $\vertex$ and let $\values<P>(\vertex<\ast>)=(X_1,X_2)$. Since
    $\values<B>(\vertex)=\true$ iff $\values<B>(\vertex<\ast>)=\false$ by
    definition and each satisfying assignment of $\vertex<\ast>$ is an
    unsatisfying assignment of $\vertex$ by definition, the claim
    follows from the induction hypothesis.
  \end{itemize}
  We have proven \Fref{theo:correctness} by induction.
\end{proof}

Let $\exampleADD$ be an attack-defense tree. We say that a
valuation~$\values[\basicEvents] \colon \basicEvents \to
\set{\true,\false}$ is a \emph{successful attack} if
$\val(\topGoal) = \true$, and \emph{unsuccessful attack}
otherwise\footnote{Dually, this can be done for defenses by assigning
  the goal~$\topGoal$ to the defender. Since we do not use player
  swapping as in~\cite{kordy2011}, there are no changes necessary to
  this or the following semantics to apply them to defenses. However,
  the valuation of the basic events might need to be modified (for
  instance, excluding costs of the attacker when computing costs of
  successful defenses). In the following, we only talk about attacks
  and do not explicitely state that our results are dually correct for
  defenses.}. We say a successful attack is \emph{minimal} if there
does not exist an basic event~$\basicEvent$ such that
$\values[\basicEvents](\basicEvent)=\true$ and the modified
valuation~$\values<\ast>$ where $\values[\basicEvents]<\ast>(\basicEvent)=\false$ result both in $\val(\topGoal) = \true$.

\begin{figure*}[!t]
  \normalsize
  \setcounter{MYtempeqncnt}{\value{equation}}
  
  \begin{equation}\label{eq:powerset}
    \begin{aligned}
    \values<\probabilityFunction>(\vertex)  =&  \sum_{x_1
        \in X_1} \prod_{\basicEvent \in x_1|_{\vertex}}
      \probabilityFunction(\basicEvent) \cdot \prod_{\basicEvent
        \in \basicEvents\setminus x_1} 1-
      \probabilityFunction(\basicEvent) =
      \\
       =&  1-\sum_{x_2|_{\vertex}
        \in X_2} \prod_{\basicEvent \in x_2|_{\vertex}}
      \probabilityFunction(\basicEvent) \cdot \prod_{\basicEvent \in
        \basicEvents \setminus x_2|_{\vertex}} 1-
      \probabilityFunction(\basicEvent)\\
    \values<\costFunction,\min>(\vertex) = &  (\min_{x_1
        \in X_1} \sum_{\basicEvent \in x_1}
      \costFunction(\basicEvent), \min_{x_2|_{\vertex}
        \in X_2} \sum_{\basicEvent \in x_2|_{\vertex}}
      \costFunction(\basicEvent))\\
    \values<\costFunction,\max>(\vertex) = & (\max_{x_1
        \in X_1} \sum_{\basicEvent \in x_1}
      \costFunction(\basicEvent),\max_{x_2|_{\vertex}
        \in X_2} \sum_{\basicEvent \in x_2|_{\vertex}}
      \costFunction(\basicEvent))\\
    \values<\delayFunction,\min>(\vertex) = &  (\min_{x_1
        \in X_1} \sum_{\basicEvent \in x_1}
      \delayFunction(\basicEvent),\min_{x_2|_{\vertex}
        \in X_2} \sum_{\basicEvent \in x_2|_{\vertex}}
      \delayFunction(\basicEvent)) \\
    \values<\delayFunction,\max>(\vertex) = & (\max_{x_1
        \in X_1} \sum_{\basicEvent \in x_1}
      \delayFunction(\basicEvent),\max_{x_1
        \in X_1} \sum_{\basicEvent \in x_1}
      \delayFunction(\basicEvent))
    \end{aligned}
  \end{equation}
  \setcounter{equation}{\value{MYtempeqncnt}}
  \hrulefill
  \vspace*{4pt}
\end{figure*}

\begin{theorem}[Correctness w.r.t. Boolean Semantics]
  \label{theo:correctnessplain}
  Let $\exampleADD$ be an attack-defense tree and
  $\values[\basicEvents]$ be a successful attack, let
  $\vertex \in \vertices$ and $\values<P>(\vertex) = (X_1,X_2)$. The
  equations in \Fref{eq:powerset} hold.
\end{theorem}

\begin{proof}
  We prove all claims by an induction on the tree structure.
  \begin{itemize}
  \item $\vertex$ is a non-triggerable basic event. For each of the
    five functions, the value assigned to the basic event is
    returned. since
    $(X_1|_{\vertex}, X_2|_{\vertex}) =
    (\set{\set{\basicEvent}},\set{\emptyset})$, the claims hold by
    definition.
  \item \emph{Induction Hypothesis:} For subtrees rooted in a
    vertex~$\vertex<\ast>$, which are predecessors of the
    vertex~$\vertex$ , the equations in \Fref{eq:powerset} hold where
    $\vertex$ is instantiated with the specific~$\vertex<\ast>$.
  \item $\typeFunction(\vertex)=\opTrigger$: The claims follow by
    definition of $\values<\probabilityFunction>$,
    $\values<\costFunction,\min>$, $\values<\costFunction,\max>$,
    $\values<\delayFunction,\min>$ and $\values<\delayFunction,\max>$
    and the induction hypothesis (since $\opTrigger$ maps to the
    identity function in all cases).
  \item $\typeFunction(\vertex)=\opNeg$: Let $\vertex<\ast>$ be the
    input of vertex. The claims for $\values<\costFunction,\min>$,
    $\values<\costFunction,\max>$, $\values<\delayFunction,\min>$ and
    $\values<\delayFunction,\max>$ follow by their definition and the
    induction hypothesis (since $\opNeg$ just swaps the components in
    these cases). For $\values<\probabilityFunction>$, it holds
    $\values<\probabilityFunction>(\vertex) =
    1-\values<\probabilityFunction>(\vertex<\ast>)$. By induction
    hypothesis, we have
    \begin{itemize}
    \item
      $1-\values<\probabilityFunction>(\vertex<\ast>) = 1 - \sum_{x_1
        \in X_1|_{\vertex<\ast>}} \prod_{\basicEvent \in x_1}
      \probabilityFunction(\basicEvent) \cdot \prod_{\basicEvent \in
        \basicEvents\setminus x_1} 1-
      \probabilityFunction(\basicEvent)$. By the definition of
      $\values<P>$ for $\opNeg$, this is equivalent to
      $1 - \sum_{x_2 \in X_2|_{\vertex<\ast>}} \prod_{\basicEvent \in
        x_2} \probabilityFunction(\basicEvent) \cdot
      \prod_{\basicEvent \in \basicEvents\setminus x_2|_{\vertex}} 1-
      \probabilityFunction(\basicEvent)$.
    \item
      $1-\values<\probabilityFunction>(\vertex<\ast>) = 1 -
      (1-\sum_{x_2 \in X_2|_{\vertex<\ast>}} \prod_{\basicEvent \in x_2}
      \probabilityFunction(\basicEvent) \cdot \prod_{\basicEvent \in
        \basicEvents \setminus x_2} 1-
      \probabilityFunction(\basicEvent))$. By definition of
      $\values<P>$ for $\opNeg$, this is equivalent to
      $\sum_{x_1 \in X_1} 1-\prod_{\basicEvent \in \basicEvents
        \setminus x_1} \probabilityFunction(\basicEvent) = \sum_{x_1
        \in X_1} \prod_{\basicEvent \in x_1}
      \probabilityFunction(\basicEvent)$
    \end{itemize}
  \item $\vertex$ is a triggerable basic event and triggered by
    $\vertex<\ast>$. Let $(X_1,X_2)=\values<P>(\vertex)$. We apply the
    induction hypothesis to $\vertex<\ast>$.
    \begin{itemize}
    \item Let $(x.y) =
      \values<\costFunction,\min>(\vertex<\ast>)$. Then,
      $\values<\costFunction,\min>(\vertex) = (x +
      \costFunction(\vertex), \min(y,\costFunction(\vertex)))$, which
      is the same as
      $(\min_{x_1 \in X_1|_{\vertex<\ast>}} \sum_{\basicEvent \in x_1}
      \costFunction(\basicEvent) + \costFunction(\vertex),
      \min(y,\costFunction(\vertex)))$. That~is
      $(\min_{x_1 \in X_1} \sum_{\basicEvent \in x_1}
        \costFunction(\basicEvent), \min_{x_2 \in X_2}
        \sum_{\basicEvent \in x_2} \costFunction(\basicEvent))$ by
      the definition of $\values<P>$ for triggerable basic events (in
      the first component, we add the vertex to all satisfying
      assignments of the input, in the second component, we have a
      union of all unsatisfying assignments, so the minimal element
      can be computed this way).
    \item Analogously, we can prove the claims for
      $\values<\costFunction,\max>$, $\values<\delayFunction,\min>$
      and $\values<\delayFunction,\max>$
    \item
      $\values<\probabilityFunction>(\vertex) =
      \values<\probabilityFunction>(\vertex<\ast>) \cdot
      \probabilityFunction(\vertex)$. By induction hypothesis, this is
      equivalent to
      $(\sum_{x_1 \in X_1|_{\vertex<\ast>}} \prod_{\basicEvent \in
        x_1} \probabilityFunction(\basicEvent) \cdot
      \prod_{\basicEvent \in \basicEvents\setminus x_1} 1-
      \probabilityFunction(\basicEvent)) \cdot
      \probabilityFunction(\vertex) = \sum_{x_1 \in X_1}
      \prod_{\basicEvent \in x_1} \probabilityFunction(\basicEvent)
      \cdot \prod_{\basicEvent \in \basicEvents\setminus x_1} 1-
      \probabilityFunction(\basicEvent)$. On the other hand, \newline
      $\displaystyle{1-\sum_{x_2 \in X_2} \prod_{\basicEvent \in x_2}
        \probabilityFunction(\basicEvent) \cdot \prod_{\basicEvent \in
          \basicEvents\setminus x_2} 1-
        \probabilityFunction(\basicEvent)}$ \newline is equivalent to
      $1-(k \cdot (1-\probabilityFunction(\vertex)) + k \cdot
      \probabilityFunction(\vertex) + (1-k) \cdot
      (1-\probabilityFunction(\vertex)))$ where
      $k=\sum_{x_2 \in X_2|_{\vertex<\ast>}} \prod_{\basicEvent \in
        x_2} \probabilityFunction(\basicEvent) \cdot
      \prod_{\basicEvent \in \basicEvents\setminus x_2} 1-
      \probabilityFunction(\basicEvent)$  by definition of $\values<P>$
      for triggerable basic events. We have $1-(k -
      k\cdot\probabilityFunction(\vertex) +
      k\cdot\probabilityFunction(\vertex)+1 - k -
      \probabilityFunction(\vertex) + k \cdot
      \probabilityFunction(\vertex)) =
      1-(1-\probabilityFunction(\vertex)+k\cdot\probabilityFunction(\vertex))
      = \probabilityFunction(\vertex) \cdot (1-k)$. By induction
      hypothesis, this is equivalent to $\probabilityFunction(\vertex)
      \cdot \values<\probability>(\vertex<\ast>)$, which matches the
      definition of $\values<\probability>$ for triggerable basic events.
    \end{itemize}
  \item $\typeFunction(\vertex)=\opAnd$: Analogously to the proof for
    triggerable basic events. However, the induction hypothesis needs
    to be used on both subtrees.
  \item $\typeFunction(\vertex)=\opOr$: Analogously to the proof for
    triggerable basic events. However, the induction hypothesis needs
    to be used for both subtrees and the proofs of the two claims need
    to be swapped.
  \end{itemize}
\end{proof}
\pagebreak
\section{Correctness of PAC Analysis}
\label{app:pac}
\begin{theorem}[Correctness of \PAC{} Analysis]
	\label{theo:pac}
	\noindent Let $\exampleADD$ be an attack-defense tree and let $\values[\basicEvents]$ be a \PAC{} valuation for basic events, i.e., for every basic event~$\basicEvent \in \basicEvents$, it holds $\values[\basicEvents](\basicEvent)$ is $(\error[\basicEvent],\prob[\basicEvent])$-\PAC{} for some $\error[\basicEvent] \in \RR_{\geq 0}$ and  $\prob[\basicEvent] \in [0,1]$. Let $(\error[\topGoal],\prob[\topGoal])$ be the error and the probability computed for the goal $\topGoal$. For any valuation $\values[\basicEvents]<\ast>$ such that $\values[\basicEvents]<\ast>(\basicEvent) \in \values[\basicEvents](\basicEvent) \pm \error[\basicEvent]$, it holds \mbox{$\values<\ast>(\topGoal) \in \values(\topGoal) \pm
		\error$}.
\end{theorem}

\begin{proof}
	
	Instead of proving the claim directly, we show that for any two \PAC{} values~$x,y \in \RR_{\geq 0}$ (with $\error[x],\prob[x]$ and $\error[y],\prob[y])$, respectively), it holds: Let $x^\ast \in x\pm\error[x]$ and $y^\ast \in y\pm\error[y]$. Then
	\begin{enumerate}
		\item $(1-x^\ast) \in (1-x)\pm\error[x]$: $x^\ast$ is at most $\error[x]$ away from $x$, i.e. the values range from
		$(1-(x+\error[x]))=(1-x)-\error[x]$ to $(1-(x-\error[x])) = (1-x) + \error[x]$.
		
		\item
		$x^\ast \cdot y^\ast \in (x \cdot y) \pm x \cdot \error[y] + y \cdot \error[x] + \error[x]\cdot\error[y]$: This claim follows by $(x \pm\error[x])\cdot(y\pm\error[y]) = y\cdot x \pm \error[x]
		\cdot y \pm y \cdot \error[x] \pm \error[x] \cdot \error[y]$.\footnote{While we can bound the error from above
			clearly to $\error[1] \cdot x_2 + x_2 \cdot \error[2] + \error[1] \cdot \error[2]$, the respective term
			$\error[1] \cdot x_2 - x_2 \cdot \error[2] - \error[1] \cdot \error[2]$ for a bound from below is not correct. Nevertheless,  we slightly over-estimate the true error from below to keep a centered interval under any circumstances and the claim still holds.}
		\item $x^\ast + y^\ast \in x+y\pm(\error[x]+\error[y])$: The claim holds since  $(x+\error[x])+(y+\error[y]) = x+y + (\error[x],\error[y])$ and  $(x-\error[x])+(y-\error[y]) = x+y - (\error[x],\error[y])$.
		\item $x^\ast + y^\ast - x^\ast \cdot y^\ast \in (x + y - x \cdot y)  \pm \error[x] + \error[y] + x \cdot \error[y] + y \cdot \error[x] + \error[x]\cdot\error[y]$: This case combines case~2 and case~3.
		\item
		$\max(x^\ast,y^\ast) \in \max(x,y) \pm \max(\error[x],\error[y])$: Clearly, $\max(x,y) + \max(\error[x],\error[y])$ is the largest value the term can attain. Since $\error[x],\error[y] \geq 0$,   $\max(x,y) - \max(\error[x],\error[y])$ is the smallest value the term can attain. Thus, the biggest error from the maximal value is given by~$\max(\error[x],\error[y])$. \footnote{Using assumptions  on the distribution of the error, more precise estimates are possible.\label{fn:assumptions}}
		\item
		$\min(x^\ast,y^\ast) \in \min(x,y) \pm \max(\error[x],\error[y])$: Clearly, $\min(x,y) - \max(\error[x],\error[y])$ is the smallest value the term can attain. Since $\error[x],\error[y] \geq 0$,    $\min(x,y) + \max(\error[x],\error[y])$ is the largest value the term can attain. Thus, the biggest error from the minimal value is given by~$\max(\error[x],\error[y])$.\footref{fn:assumptions}
	\end{enumerate}
	The result of \Fref{theo:pac} then follows from a straightforward induction on the structure of the tree.
\end{proof}




\newpage
\addcontentsline{toc}{chapter}{Bibliography}
\printbibliography

\end{document}